\documentclass[final,hidelinks,onefignum,onetabnum]{siamart251216}

\usepackage{multirow}
\usepackage{booktabs}
\usepackage{graphicx,epstopdf} 
\usepackage[caption=false]{subfig} 

\usepackage{lipsum}
\usepackage{amsfonts}
\usepackage{graphicx}
\usepackage{epstopdf}
\usepackage{algorithmic}
\ifpdf
  \DeclareGraphicsExtensions{.eps,.pdf,.png,.jpg}
\else
  \DeclareGraphicsExtensions{.eps}
\fi

\newsiamremark{remark}{Remark}
\newsiamremark{hypothesis}{Hypothesis}
\crefname{hypothesis}{Hypothesis}{Hypotheses}
\newsiamthm{claim}{Claim}
\newsiamremark{fact}{Fact}
\crefname{fact}{Fact}{Facts}

\headers{Hybrid Fixed-Stress Transferable Neural Network}{Z. Shen, L. Zhu, and M. Cai}

\title{HFS-TransNet: A Hybrid Fixed-Stress Transferable Neural Network for Quasi-Static Biot Poroelasticity\thanks{Submitted to the editors DATE.
\funding{This work was supported by the National Natural Science Foundation of China (grants
11871091 and 12571580).}}}

\author{Zhequan Shen\thanks{School of Mathematical Sciences, Beihang University, Beijing 102206, China 
  (\email{zhequanshen@buaa.edu.cn}).}
\and Liyong Zhu\thanks{School of Mathematical Sciences, Beihang University, Beijing 102206, China
  (\email{liyongzhu@buaa.edu.cn}).}
\and {Mingchao Cai\thanks{Department of Mathematics, Morgan State University, Baltimore, MD 21251, USA
  (\email{Mingchao.Cai@morgan.edu}).}}}

\usepackage{amsopn}

\ifpdf
\hypersetup{
  pdftitle={HFS-TransNet: A Hybrid Fixed-Stress Transferable Neural Network for Quasi-Static Biot Poroelasticity},
  pdfauthor={Z. Shen, L. Zhu, and M. Cai}
}
\fi

\begin{document}

\maketitle

\begin{abstract}
The quasi-static Biot system, which governs coupled fluid-solid interactions in poroelastic media, poses significant computational challenges due to its strong coupling and multiscale nature.
To address these, we propose a hybrid fixed-stress transferable neural network (HFS-TransNet) method.
In HFS-TransNet, a data-driven TransNet module first approximates observed data on an initialization domain to supply a high-quality starting point for the subsequent physics-informed stage, which then drives fixed-stress splitting iterations via TransNet and admits an error bound.
After each iteration, a greedy update strategy on a validation domain prevents model degradation and a relative tolerance criterion governs convergence termination.
HFS-TransNet thereby inherits the robustness of the fixed-stress splitting scheme, the flexibility of data-driven modeling, and the transferability of TransNet, effectively overcoming the strong coupling and locking instability inherent in the Biot system.
Ablation studies and comparisons with FS-FEM and FS-PINN confirm HFS-TransNet's superior performance in capturing multiscale coupled physics across varying physical parameters and boundary conditions.
By bridging classical decoupled iterative schemes with hybrid scientific machine learning, HFS-TransNet offers a novel perspective for complex poroelastic simulations, as demonstrated by its successful application to brain edema simulations.
\end{abstract}

\begin{keywords}
Biot model, hybrid SciML, fixed-stress, robustness, transferability
\end{keywords}

\begin{MSCcodes}
65M15, 68T07, 74F10
\end{MSCcodes}

\section{Introduction}
\label{sec: intro}
Biot consolidation model \cite{biot1941general, showalter2000diffusion, vzenivsek1984existence} in poroelasticity describes the coupled interaction between fluid flow and the porous solid skeleton, governed by a system of partial differential equations linking fluid pressure and solid displacement. 
This model has broad applications in geotechnical \cite{huang2022generalized}, biomechanical \cite{chen2023physical, ju2020parameter, wang2025mixed}, and reservoir \cite{kim20111591} engineering.
Due to its multiscale nature, strong coupling, and sensitivity to physical parameters, classic two-field finite element methods (FEM) often suffer from elastic locking and pressure oscillations \cite{lee2017parameter, phillips2007coupling, yi2017study}.
Therefore, three- or four-field reformulation using inf-sup stable finite element pairs \cite{chu2026block, feng2018analysis, mora2026118447, oyarzua2016locking} have been proposed.
Despite this, these methods may incur substantial computational cost.
To overcome these challenges, the fixed-stress splitting scheme \cite{kim20111591, kim2011stability, mikelic2013convergence, storvik2019optimization} decomposes the Biot system into two iteratively solved subproblems, with its convergence rigorously established.

Unlike traditional grid-based methods, neural network methods have witnessed rapid development owing to their mesh-free nature.
Purely data-driven neural network methods have been successfully employed to discover unknown models from observed data \cite{chen2025due}.
When applied to solving well-defined physical models, although these methods excel in regions rich with observation, they often struggle to capture the underlying physics, especially when required to extrapolate beyond the observed data range \cite{liverani2025hyco}.
To address this issue, physics-informed neural networks (PINN) \cite{raissi2019physics} have been developed.
To enhance training stability in poroelasticity simulations, the fixed-stress method has been integrated into PINN, either by embedding it directly into the training process \cite{haghighat2022physics} or by coupling two independent networks for displacement and pressure through an iterative algorithm (FS-PINN) \cite{cai2023combination}.
Although these physics-informed methods enable extrapolation beyond the observed domain, they tend to suffer from instability, largely due to equation stiffness and data scarcity \cite{liverani2025hyco}.

To balance the trade-off between the flexibility of purely data-driven models and the interpretability of physics-informed models, the hybrid-cooperative (HYCO) framework proposed in \cite{liverani2025hyco} iteratively integrates the two approaches through a mutual regularization mechanism, enabling accurate recovery of solutions under ill-posed conditions.
Furthermore, a novel scientific machine learning (SciML) paradigm proposed in \cite{Quarteroni2025com} leverages the complementary strengths of data-driven methods in discovering complex nonlinear patterns and physics-informed models in providing mathematical interpretability, achieving significant success in the simulation of human cardiac function.
However, the above-mentioned neural network methods usually rely on deep architectures and empirical hyperparameters, leading to high optimization costs and substantial manual effort in hyperparameter selection.

To reduce the computational cost associated with deep architectures, shallow neural network methods with pretrained hidden-layer parameters are introduced in \cite{cai2020deep, chen2022bridging, dong2021local, dong2022computing, li2025local, shang2025overlapping, weng2026deep}.
Furthermore, to alleviate the burden of hyperparameter tuning caused by low interpretability, the transferable neural network (TransNet) has been proposed in \cite{zhang2024transferable}.
TransNet is a two-layer neural network with pretrained hidden-layer neurons possessing explicit geometric meanings.
The numerical results in \cite{cheng2026generalized, lu2025multiple, shen2026matched, xu2026weak, yan2026time, zhang2024transferable} verify the excellent transferability of TransNet. 
However, when applied to the Biot model, it still struggles to capture coupled physical effects, especially in nearly incompressible regimes as the Poisson's ratio $\nu \rightarrow 0.5$.

Inspired by FS-PINN, TransNet and HYCO, we propose a hybrid fixed-stress transferable neural network (HFS-TransNet) for solving Biot model.
By leveraging the interaction between a purely data-driven TransNet module and a physics-informed TransNet module embedded with fixed-stress splitting, this method effectively overcomes the strong coupling and parameter instabilities of the Biot system.
The main contributions of our work are summarized as follows.
\begin{itemize}
    \item We propose a novel framework that bridges traditional decoupled iterative algorithms with hybrid scientific machine learning approaches, offering a new perspective for addressing complex engineering problems.
    \item We integrate the classical fixed-stress splitting iteration into TransNet, providing error estimates to ensure the physics-informed module inherits the stability of the iterative scheme while retaining the accuracy of the network.
    \item We design a TransNet-based data-driven module that provides high-quality initialization for the physics-informed stage and adaptively updates the iteration weights, thereby effectively reducing the computational cost.
    \item We conduct extensive ablation studies to confirm the effectiveness of the HFS-TransNet, and comprehensive comparative experiments to demonstrate its superiority in accuracy, efficiency, and robustness, supporting its applicability in Barry--Mercer problem, Mandel problem and brain edema simulations.
\end{itemize}

The rest of the paper is organized as follows. 
\Cref{sec: biot} introduces the Biot model. 
In \Cref{sec: met}, the HFS-TransNet is developed with an error estimate. 
\Cref{sec: ne} presents numerical results on a set of benchmark experiments and brain edema simulations.
Finally, \Cref{sec: ne_con} concludes this work and outlines future research directions.

\section{Biot model}
\label{sec: biot}
In this paper, we consider the following quasi-static Biot system in a bounded domain $\Omega \subset \mathbb{R}^2$ over the time interval $(0, T]$:
\begin{align}
    \partial_t \left( \frac{1}{M} p + \alpha \nabla \cdot \boldsymbol{u} \right) - \nabla \cdot K \nabla p = f \quad \text{in } \Omega \times (0, T], \label{equ: fluid mass conservation} \\
    -\nabla \cdot \sigma(\boldsymbol{u}) + \alpha \nabla p = \boldsymbol{g} \quad \text{in } \Omega \times (0, T]. \label{equ: solid force equilibrium}
\end{align}
\Cref{equ: fluid mass conservation} expresses fluid mass conservation, while \Cref{equ: solid force equilibrium} expresses solid force equilibrium.
In the system, $\boldsymbol{u}=(u, v)^\top$ and $p$ are the target functions, which denote the displacement vector of the solid phase and pressure of the fluid phase respectively.
Meanwhile, $f$ is a source or sink term, $\boldsymbol{g}$ is the body force, $M > 0$ is the Biot modulus, $\alpha$ is the Biot-Willis constant, $K = \kappa/\mu_f$ is the hydraulic conductivity with the permeability $\kappa > 0$ and the fluid viscosity $\mu_f$.
\begin{equation}
    \sigma(\boldsymbol{u}) := 2\mu\varepsilon(\boldsymbol{u}) + \lambda (\nabla\cdot\boldsymbol{u}) \boldsymbol{I}, \quad
    \varepsilon(\boldsymbol{u}) := \frac{1}{2}\left(\nabla\boldsymbol{u} + \nabla\boldsymbol{u}^T\right), \nonumber
\end{equation}
where $\lambda$ and $\mu$ are Lamé constants computed from the Young's modulus $E$ and the Poisson's ratio $\nu$:
\begin{equation}
    \lambda = \frac{E\nu}{(1+\nu)(1-2\nu)}, \quad \mu = \frac{E}{2(1+\nu)}. \nonumber
\end{equation}

To define a well-posed mathematical problem, the above system needs appropriate boundary and initial conditions. Without loss of generality, we consider mixed Neumann and Dirichlet boundary conditions following \cite{ju2020parameter}:
\begin{equation}\label{equ: boundary condition}
    \left\{
    \begin{aligned}
        \boldsymbol{u} &= \boldsymbol{u}_d & \quad \text{on } \Gamma_{\boldsymbol{u}}^D \times (0,T], \\
        \sigma(\boldsymbol{u})\boldsymbol{n} - \alpha p\boldsymbol{n} &= \boldsymbol{h} & \quad \text{on } \Gamma_{\boldsymbol{u}}^N \times (0,T], \\
        p &= p_d & \quad \text{on } \Gamma_p^D \times (0,T], \\
        K\nabla p \cdot \boldsymbol{n} &= q & \quad \text{on } \Gamma_p^N \times (0,T],
    \end{aligned}
    \right.
\end{equation}
where $\Gamma_{\boldsymbol{u}}^D, \Gamma_{\boldsymbol{u}}^N$ and $\Gamma_p^D, \Gamma_p^N$ are the Dirichlet and Neumann boundaries for $\boldsymbol{u}$ and $p$ respectively, $\partial\Omega = \Gamma_{\boldsymbol{u}}^D \cup \Gamma_{\boldsymbol{u}}^N = \Gamma_p^D \cup \Gamma_p^N$. 
Additionally, the initial conditions are
\begin{equation}\label{equ: initial condition}
    \boldsymbol{u}(0) = \boldsymbol{u}_0, \quad p(0) = p_0, \quad \text{in}\, \Omega \times \{0\}.
\end{equation}

\section{Methodology}
\label{sec: met}
In this section, we briefly review the transferable neural network (TransNet) and identify its bottlenecks in solving Biot systems, subsequently presenting our HFS-TransNet framework along with its numerical analysis.

\subsection{Network architecture}
\label{sec: pre_TransNet}
The TransNet \cite{zhang2024transferable} is a fully connected two-layer neural network (one hidden layer and one output layer) with pretrained hidden-layer neurons for solving PDEs of the following general formulation:
\begin{equation}\label{equ: TransNet_PDE}
    \left\{
    \begin{aligned}
        & \mathcal{L}(\boldsymbol{u}) = f, \quad \boldsymbol{x} \in \Omega\times(0,T],\\
        & \mathcal{B}(\boldsymbol{u}) = g, \quad \boldsymbol{x} \in \partial\Omega\times(0,T],\\
        & \mathcal{I}(\boldsymbol{u}) = h, \quad \boldsymbol{x} \in \Omega\times\{0\},\\
    \end{aligned}
    \right.
\end{equation}
where $\Omega$ denotes the spatial domain in $\mathbb{R}^d$. 
As shown in \cref{fig: TransNet}, the computational procedure of TransNet consists of two stages: generating pretrained hidden-layer parameters, and optimizing the parameters of the output layer.
\begin{figure}[tbhp]
    \centering
	\includegraphics[width=0.95\textwidth]{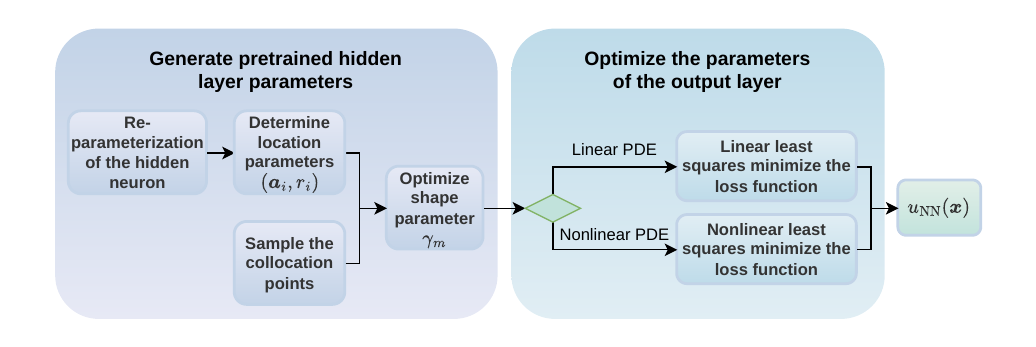}
	\vspace{-0.5cm}
	\caption{The flowchart of the transferable neural network method for solving PDEs.}
	\label{fig: TransNet}
\end{figure}

In the first stage, the pretrained model is the neural feature space $\mathcal{P}_{\mathrm{NN}}$, expanded by hidden-layer neurons (or called neural basis functions)  $\left\{\sigma\left(\boldsymbol{w}_i^{\top} \boldsymbol{x}+b_i\right)\right\}_{i=1}^I$:
\begin{equation}\label{equ: TransNet_NFS}
	\mathcal{P}_{\mathrm{NN}}=\operatorname{span}\left\{\sigma\left(\boldsymbol{w}_1^{\top} \boldsymbol{x}+b_1\right), \ldots, \sigma\left(\boldsymbol{w}_I^{\top} \boldsymbol{x}+b_I\right)\right\},
\end{equation} 
where $I$ is the number of hidden-layer neurons, $\sigma(\cdot)$ is the activation function, and $\boldsymbol{w}_i\in \mathbb{R}^{d+1}, b_i\in \mathbb{R}$ is the weight and bias of the $i$-th hidden-layer neuron. 
To achieve a uniform distribution of neurons in $\mathcal{P}_{\mathrm{NN}}$ within the ball $B_R(\boldsymbol{x}_c)$ (centered at the point $\boldsymbol{x}_c \in \mathbb{R}^{d+1}$ with a radius $R \in \mathbb{R}_+$), the initial step involves the re-parameterization of the hidden-layer neuron as:
\begin{equation}
	\sigma\left(\boldsymbol{w}_i^{\top} \boldsymbol{x}+b_i\right)
	= \sigma\left(\gamma_i\left(\boldsymbol{a}_i^{\top} (\boldsymbol{x}-\boldsymbol{x}_c)+r_i\right)\right), \nonumber
\end{equation}
where $\boldsymbol{a}_i \in \mathbb{R}^{d+1}, \|\boldsymbol{a}_i\|_2=1$ is the unit normal vector of the partition hyperplane corresponding to the $i$-th hidden-layer neuron, i.e., $\boldsymbol{w}_i^{\top} \boldsymbol{x}+b_i=0$.
Here, $r_i \in \mathbb{R}$, $|r_i|$ represents the distance from the point $\boldsymbol{x}_c$ to the partition hyperplane, $(\boldsymbol{a}_i, r_i)$ are named location parameters, while $\gamma_i \in \mathbb{R}_+$ is the shape parameter. 
In this work, we take $\sigma(\cdot)=\tanh(\cdot)$, which is widely used for solving PDEs. 
The shape parameter $\gamma_i$ controls the transition layer width of the activation function near the partition hyperplane. 
{\em The location parameters $\left\{\boldsymbol{a}_i\right\}_{i=1}^I$ are i.i.d., uniformly distributed on the $(d+1)$-dimensional unit sphere and $\left\{r_i\right\}_{i=1}^I$ are i.i.d., uniformly distributed in $[0,R]$.}
For simplicity, it is usually assumed that all hidden-layer neurons share the same shape parameter, i.e., $\gamma_i=\gamma$. The next step is optimizing the shape parameter $\gamma$, which can be achieved by some simple line search algorithms (see \cite{lu2025multiple} for details). 
At this stage, the construction of $\mathcal{P}_{\mathrm{NN}}$ is completed. 
To simplify notation, we assume that solution $\boldsymbol{u}$ in \cref{equ: TransNet_PDE} is a scalar function, then the network solution $u_\mathrm{NN}\in \mathcal{P}_\mathrm{NN}$ is expressed as
\begin{equation}\label{equ: u_NN}
	u_{\mathrm{NN}}(\boldsymbol{x})=\sum_{i=1}^I \alpha_i \sigma\left(\boldsymbol{w}_i^{\top} \boldsymbol{x}+b_i\right),
\end{equation}
where $\alpha_1, \alpha_2, \dots, \alpha_I$ are the weights of the output layer.

In the second stage, the parameters of the hidden layer are frozen, and the weights of the output layer, that is, $\alpha_1, \dots, \alpha_I$ in \cref{equ: u_NN} can be obtained through least-squares solvers, which are obtained by minimizing the PDE-residual loss function: 
\begin{equation}\label{equ: TransNet_loss}
	{\mathrm{Loss}} = \lambda_1\|\mathcal{L}(u(\boldsymbol{x})) - \mathcal{L}(u_{\mathrm{NN}}(\boldsymbol{x}))\|_2^2 
    + \lambda_2\|\mathcal{B}(u(\boldsymbol{x})) - \mathcal{B}(u_{\mathrm{NN}}(\boldsymbol{x})) \|_2^2
    + \lambda_3\|\mathcal{I}(u(\boldsymbol{x})) - \mathcal{I}(u_{\mathrm{NN}}(\boldsymbol{x})) \|_2^2,
\end{equation}
where $\lambda_1$, $\lambda_2$, and $\lambda_3$ are adaptively calculated by a normalization approach proposed in \cite{lu2025multiple}.
\cref{alg: TransNet} summarizes the main steps of the TransNet method. 

\begin{algorithm}
	\caption{TransNet method for solving PDE problems \cite{zhang2024transferable}}
	\label{alg: TransNet}
    \begin{algorithmic}[1]
	\STATE{\textbf{Input:} The spatial-temporal domain $\Omega\times(0,T]$, the point $\boldsymbol{x}_c$, the radius $R$, the shape parameter $\gamma$, the number of hidden-layer neurons $I$.}
    
	\STATE{\textbf{Output:} The TransNet solution $u_{\mathrm{NN}}$.}
	
    \STATE{Generate the location parameters $\left\{(\boldsymbol{a}_i, r_i)\right\}_{i=1}^{I}$ in $\Omega\times(0,T]$ as the process described in \cite{lu2025multiple} and form the set of hidden-layer neurons $\{\sigma\left(\gamma\left(\boldsymbol{a}_i^{\top} (\boldsymbol{x}-\boldsymbol{x}_c)+r_i\right)\right)\}_{i=1}^{I}$ (uniformly distributed in $\Omega\times(0,T]$).}
    
    \STATE{Generate collocation points in $\Omega\times(0,T]$ by uniform sampling.}

    \STATE{Optimize the output-layer parameters $\boldsymbol{\alpha}$ by minimizing the loss function \cref{equ: TransNet_loss} associated with the PDE problem \cref{equ: TransNet_PDE} over the collocation points using a specified least-squares solver.}

    \RETURN{$u_{\mathrm{NN}}$ according to \cref{equ: u_NN}.}
    \end{algorithmic}
\end{algorithm}

Extensive numerical experiments  indicate that TransNet has much higher accuracy (numerically exponential-type convergence with respect to the number of neurons) and lower computing cost compared to PINN for various PDE problems with relatively smooth solutions \cite{zhang2024transferable}. 
However, numerical experiments in \cref{tab: case2_ablation_nu} reveal TransNet's instability when solving Biot model.
To address this limitation, we subsequently propose HFS-TransNet.

\subsection{Hybrid fixed-stress transferable neural network}
\label{sec: met_hfs} 
The HFS-TransNet comprises two modules. One is a physics-informed module (also named FS-TransNet), which integrates the fixed-stress splitting scheme into the TransNet. The other is a data-driven module designed to initialize iterative weights and to prevent model degradation. Detailed descriptions of these two modules are given below.

\subsubsection{Physics-informed module}
\label{sec: met_hfs_pi}
Motivated by the fixed-stress splitting method, we propose the fixed-stress transferable neural network (FS-TransNet), a physics-informed approach that combines the unconditional stability of the fixed-stress splitting scheme \cite{kim2011stability} with the high accuracy and efficiency of TransNet \cite{zhang2024transferable}. 
Given an initial guess $(\boldsymbol{u}_{\mathrm{NN}}^0, p_{\mathrm{NN}}^0)=\boldsymbol{0}$, the FS-TransNet procedure is described as follows. 
\newline
\textbf{Step 1}: Given $\boldsymbol{u}_{\mathrm{NN}}^n$ and $p_{\mathrm{NN}}^n$, we solve for $p_{\mathrm{NN}}^{n+1}$ to approximate $p^{n+1}$, satisfying
\begin{align}
    \left( \frac{1}{M} + \beta_{FS} \right) \partial_t p^{n+1} - \nabla \cdot K \nabla p^{n+1} &= f - \alpha \nabla \cdot \partial_t \boldsymbol{u}_{\mathrm{NN}}^n + \beta_{FS} \partial_t p_{\mathrm{NN}}^n \quad &\text{in } \Omega \times (0, T], \label{equ: Step 1 PDE NN} \\
    p^{n+1} &= p_d \quad &\text{on } \Gamma_p^D\times (0, T], \label{equ: Step 1 Dirichlet NN} \\
    K\nabla p^{n+1} \cdot \boldsymbol{n} &= q \quad &\text{on } \Gamma_p^N\times (0, T], \label{equ: Step 1 Neumann NN} \\
    p^{n+1} &= p_0 \quad &\text{in } \Omega \times \{0\}, \label{equ: Step 1 Initial NN}
\end{align}
\textbf{Step 2}: Using $p_{\mathrm{NN}}^{n+1}$, we solve for $\boldsymbol{u}_{\mathrm{NN}}^{n+1}$ to approximate $\boldsymbol{u}^{n+1}$, satisfying
\begin{align}
    -\nabla \cdot \sigma(\boldsymbol{u}^{n+1}) &= \boldsymbol{g} - \alpha \nabla p_{\mathrm{NN}}^{n+1} \quad &\text{in } \Omega \times (0, T], \label{equ: Step 2 PDE NN} \\
    \boldsymbol{u}^{n+1} &= \boldsymbol{u}_d \quad &\text{on } \Gamma_{\boldsymbol{u}}^D\times (0, T], \label{equ: Step 2 Dirichlet NN} \\
    \sigma(\boldsymbol{u}^{n+1})\boldsymbol{n} &= \boldsymbol{h} + \alpha p_{\mathrm{NN}}^{n+1}\boldsymbol{n} \quad &\text{on } \Gamma_{\boldsymbol{u}}^N\times (0, T], \label{equ: Step 2 Neumann NN} \\
    \boldsymbol{u}^{n+1} &= \boldsymbol{u}_0 \quad &\text{in } \Omega \times \{0\}, \label{equ: Step 2 Initial NN}
\end{align}

Notably, the three fields share one frozen hidden-feature map but have separate output layers.
Specifically, the $x$- and $y$-components of displacement, $u_{\mathrm{NN}}(\boldsymbol{x}; \boldsymbol{\alpha}_u)$ and $v_{\mathrm{NN}}(\boldsymbol{x}; \boldsymbol{\alpha}_v)$, as well as the pressure $p_{\mathrm{NN}}(\boldsymbol{x}; \boldsymbol{\alpha}_p)$ share the same set of pretrained hidden-layer parameters.
As shown in \cref{fig: FS-TransNet}, each iteration of the FS-TransNet method involves two steps: using $\boldsymbol{\alpha}_u^{n}, \boldsymbol{\alpha}_v^{n}, \boldsymbol{\alpha}_p^{n}$ to update $\boldsymbol{\alpha}_p^{n+1}$, and then using $\boldsymbol{\alpha}_p^{n+1}$ to update $\boldsymbol{\alpha}_u^{n+1}, \boldsymbol{\alpha}_v^{n+1}$.
\cref{alg: FS-TransNet} summarizes the implementation details of the physics-informed FS-TransNet for Biot model.
\begin{figure}[tbhp]
    \centering
	\includegraphics[width=.9\textwidth]{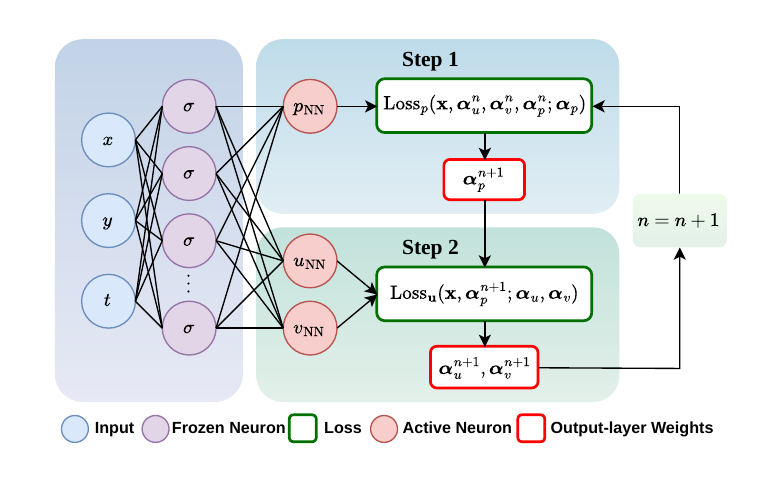}
	\vspace{-.8cm}
	\caption{The flowchart of the physics-informed FS-TransNet method.}
	\label{fig: FS-TransNet}
\end{figure}
\begin{algorithm}
	\caption{FS-TransNet method for Biot model}
	\label{alg: FS-TransNet}
    \begin{algorithmic}[1]
	\STATE{\textbf{Input}: The spatial-temporal domain $\Omega\times(0,T]$, the point $\boldsymbol{x}_c$, the radius $R$, the shape parameter $\gamma$, the number of hidden-layer neurons $I$, the maximum number of iterations $L$.
	}
	\STATE{\textbf{Output}: The FS-TransNet solution $u_{\mathrm{NN}}, v_{\mathrm{NN}}, p_{\mathrm{NN}}$.}

    \STATE{Generate the set of hidden-layer neurons $\bigl\{\sigma\!\left(\gamma\!\left(\boldsymbol{a}_i^{\top} (\boldsymbol{x}-\boldsymbol{x}_c)+r_i\right)\right)\bigr\}_{i=1}^{I}$ (uniformly distributed in $\Omega\times(0,T]$).}

    \STATE{Generate collocation points in $\Omega\times(0,T]$.}
    
    \STATE{Initialize the output-layer weights $\boldsymbol{\alpha}_u^0=\boldsymbol{\alpha}_v^0=\boldsymbol{\alpha}_p^0=\boldsymbol{0}$.}
    
    \FOR{$n = 0:(L-1)$}

    \STATE{\textbf{Step 1}: Update $\boldsymbol{\alpha}_p^{n+1}$ by minimizing the loss \cref{equ: TransNet_loss} associated with system \cref{equ: Step 1 PDE NN}--\cref{equ: Step 1 Initial NN} over the collocation points in $\Omega\times(0,T]$ using $\boldsymbol{\alpha}_u^{n}, \boldsymbol{\alpha}_v^{n}, \boldsymbol{\alpha}_p^{n}$.}

    \STATE{\textbf{Step 2}: Update $\boldsymbol{\alpha}_u^{n+1}, \boldsymbol{\alpha}_v^{n+1}$ by minimizing the loss \cref{equ: TransNet_loss} associated with system \cref{equ: Step 2 PDE NN}--\cref{equ: Step 2 Initial NN} over the collocation points in $\Omega\times(0,T]$ using $\boldsymbol{\alpha}_p^{n+1}$.}

    \ENDFOR
    
    \RETURN{$u_{\mathrm{NN}}=u_{\mathrm{NN}}^L, v_{\mathrm{NN}}=v_{\mathrm{NN}}^L, p_{\mathrm{NN}}=p_{\mathrm{NN}}^L$ according to \cref{equ: u_NN}.}
    \end{algorithmic}
\end{algorithm}

Rooted in the governing PDEs, physics-informed methods provide greater interpretability than data-driven approaches and demonstrate a superior ability to extrapolate beyond the observed data domain, which results in better generalization \cite{liverani2025hyco}.
Consequently, we employ FS-TransNet as our physics-informed module. 
However, FS-TransNet faces two distinct challenges. 
First, the use of zero initialization in \cref{alg: FS-TransNet} is overly simplistic. A proper initialization could substantially reduce the number of iterations required for convergence, thereby lowering the computational cost. 
Second, the iteration count depends on a fixed hyperparameter $L$. A proper iterative update mechanism combined with a relative tolerance criterion eliminates the reliance on manual hyperparameter tuning, thereby further enhancing the generalization capability of the method.
In the following, we introduce a data-driven module to address these two challenges.

\subsubsection{Data-driven module}
\label{sec: met_hfs_dd}
In this section, the original zero initialization in the physics-informed module is replaced by a data-driven initialization. Meanwhile, a greedy update strategy based on data-driven validation is employed for the iteration weights, with a relative tolerance serving as the termination criterion.
First, the computational domains for the two modules are illustrated in \cref{fig: Domain}. 
The observed data in the data-driven module are divided into initialization and validation domains, while the physics-informed module covers both observed and unobserved data.
Next, we provide a detailed description of the TransNet-based data-driven module including its initialization and validation components.
\begin{figure}[tbhp]
    \centering
	\includegraphics[width=.9\textwidth]{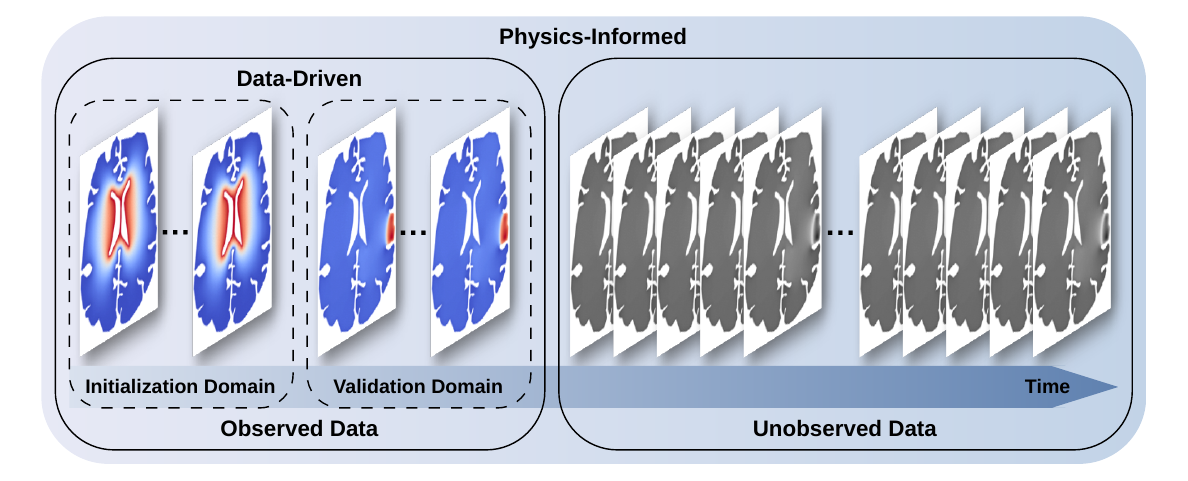}
	\vspace{-.4cm}
	\caption{Computational domains of the physics-informed and data-driven modules in HFS-TransNet.}
	\label{fig: Domain}
\end{figure}

During the initialization, to obtain the appropriate initial parameters $\{\boldsymbol{\alpha}_u^0, \boldsymbol{\alpha}_v^0, \boldsymbol{\alpha}_p^0\}$ in \cref{alg: FS-TransNet}, TransNet is utilized to approximate the observed data within the initialization domain.
The data-driven TransNet follows the same procedure as described in Section \ref{sec: pre_TransNet}, except that the loss function \cref{equ: TransNet_loss} is replaced by \cref{equ: TransNet_loss_data}: 
\begin{equation}\label{equ: TransNet_loss_data}
    {\mathrm{Loss}_\phi} = \|\phi_{\mathrm{init}}^{\mathrm{obs}}(\boldsymbol{x}) - \phi_{\mathrm{NN}}^0(\boldsymbol{x})\|_2^2, \quad \phi \in \{u, v, p\},
\end{equation}
where $\phi_{\mathrm{init}}^{\mathrm{obs}}(\boldsymbol{x})$ denotes the observed data in the initialization domain.

In the validation, a greedy update strategy is employed to prevent model degradation. The network approximation error is defined as follows:
\begin{equation}\label{equ: NN_error}
    \Delta \phi_{\mathrm{NN}}^{n+1}= \|\phi_{\mathrm{val}}^{\mathrm{obs}}(\boldsymbol{x}) - \phi_{\mathrm{NN}}^{n+1}(\boldsymbol{x})\|_2^2, \quad \phi \in \{u, v, p\}, 
\end{equation}
where $\phi_{\mathrm{val}}^{\mathrm{obs}}(\boldsymbol{x})$ denotes the observed data in the validation domain.
If $\Delta \phi_{\mathrm{NN}}^{n+1} < \Delta \phi_{\mathrm{NN}}^{n}$, the iteration is considered valid, and $\boldsymbol{\alpha}_\phi^{n+1}$ is used to overwrite $\boldsymbol{\alpha}_\phi^{n}$. 
Otherwise, the previous parameters $\boldsymbol{\alpha}_\phi^{n}$ are retained.
Furthermore, a relative tolerance $\mathrm{Tol}$ is introduced to determine iterative convergence. 
The iteration terminates once the condition $\| \boldsymbol{\alpha}_\phi^{n+1} - \boldsymbol{\alpha}_\phi^n \|_2 / \| \boldsymbol{\alpha}_\phi^n \|_2 < \mathrm{Tol}$ is satisfied for all $\phi \in \{u, v, p\}$.

Having established both modules, the overall workflow of the HFS-TransNet algorithm is illustrated in \cref{fig: HFS-TransNet}.
First, the data-driven module is employed to initialize the output-layer weights $\boldsymbol{\alpha}_u^0,\boldsymbol{\alpha}_v^0,\boldsymbol{\alpha}_p^0$ corresponding to the scalar variables $u, v, p$ respectively. 
Subsequently, the physics-informed module commences the fixed-stress splitting iteration. 
After each iteration, a validation step based on the data-driven module is carried out to perform a greedy update, after which the tolerance $\mathrm{Tol}$ is used to determine whether the iterative process should terminate.
\cref{alg: HFS-TransNet} summarizes the main steps of the HFS-TransNet for Biot model.
\begin{figure}[tbhp]
    \centering
	\includegraphics[width=.9\textwidth]{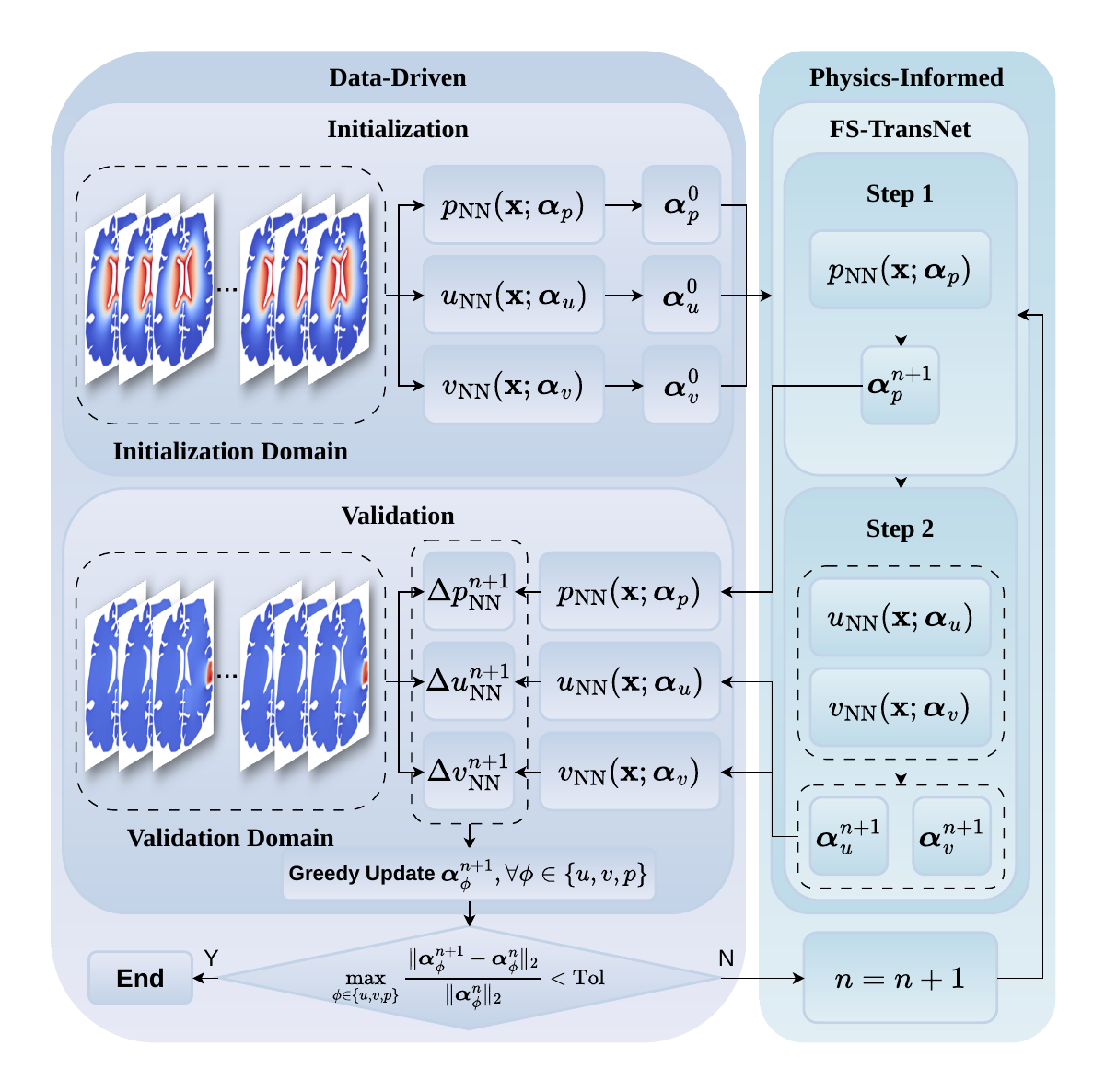}
	\vspace{-.8cm}
	\caption{The flowchart of the hybrid fixed-stress transferable neural network (HFS-TransNet).}
	\label{fig: HFS-TransNet}
\end{figure}
\begin{remark}
    In practice, if observational data are unavailable, the initialization can be replaced by a physics-informed approach, and the validation stage is bypassed, thereby reducing the reliance on labeled data, as demonstrated in Section \ref{sec: ne_md}.
\end{remark}

\begin{algorithm}
	\caption{HFS-TransNet method for Biot model}
	\label{alg: HFS-TransNet}
    \begin{algorithmic}[1]
	\STATE{\textbf{Input}: The spatial-temporal domain $\Omega\times(0,T]$, the initialization domain $\Omega\times(0,T_1]$, the validation domain $\Omega\times(T_1,T_2]$, the point $\boldsymbol{x}_c$, the radius $R$, the shape parameter $\gamma$, the number of hidden-layer neurons $I$, the maximum number of iterations $L$, the tolerance $\mathrm{Tol}$.
	}
	\STATE{\textbf{Output}: The HFS-TransNet solution $u_{\mathrm{NN}}, v_{\mathrm{NN}}, p_{\mathrm{NN}}$.}

    \STATE{Generate the set of hidden-layer neurons $\{\sigma\left(\gamma\left(\boldsymbol{a}_i^{\top} (\boldsymbol{x}-\boldsymbol{x}_c)+r_i\right)\right)\}_{i=1}^{I}$ (uniformly distributed in $\Omega\times(0,T]$).}

    \STATE{Generate collocation points in $\Omega\times(0,T_1]$, $\Omega\times(T_1,T_2]$, and $\Omega\times(0,T]$.}

    \STATE{Optimize the output-layer parameters $\boldsymbol{\alpha}_u^0, \boldsymbol{\alpha}_v^0, \boldsymbol{\alpha}_p^0$ by minimizing the data‑driven loss \cref{equ: TransNet_loss_data} over the collocation points in $\Omega\times(0,T_1]$.}
    
    \FOR{$n = 0:(L-1)$}
    \STATE{\textbf{Step 1}: Update $\boldsymbol{\alpha}_p^{n+1}$ by minimizing the loss \cref{equ: TransNet_loss} associated with system \cref{equ: Step 1 PDE NN}--\cref{equ: Step 1 Initial NN} over the collocation points in $\Omega\times(0,T]$ using $\boldsymbol{\alpha}_u^{n}, \boldsymbol{\alpha}_v^{n}, \boldsymbol{\alpha}_p^{n}$.}

    \STATE{\textbf{Step 2}: Update $\boldsymbol{\alpha}_u^{n+1}, \boldsymbol{\alpha}_v^{n+1}$ by minimizing the loss \cref{equ: TransNet_loss} associated with system \cref{equ: Step 2 PDE NN}--\cref{equ: Step 2 Initial NN} over the collocation points in $\Omega\times(0,T]$ using $\boldsymbol{\alpha}_p^{n+1}$.}

    \STATE{Compute the network approximation errors $\Delta \phi_{\mathrm{NN}}^{n+1}$ ($\phi \in \{u, v, p\}$) in \cref{equ: NN_error} over the collocation points in $\Omega\times(T_1,T_2]$ using $\boldsymbol{\alpha}_u^{n+1}, \boldsymbol{\alpha}_v^{n+1}, \boldsymbol{\alpha}_p^{n+1}$.}

    \FOR{$\phi \in \{u, v, p\}$}
    \IF{$\Delta \phi_{\mathrm{NN}}^{n+1} \geq \Delta \phi_{\mathrm{NN}}^{n}$}
        \STATE{$\boldsymbol{\alpha}_\phi^{n+1} \leftarrow \boldsymbol{\alpha}_\phi^{n}$.}
    \ENDIF
    \ENDFOR
    \IF{$\max_{\phi \in \{u, v, p\}} \| \boldsymbol{\alpha}_\phi^{n+1} - \boldsymbol{\alpha}_\phi^n \|_2 / \| \boldsymbol{\alpha}_\phi^n \|_2 < \mathrm{Tol}$}
        \STATE{\textbf{break}.}
    \ENDIF
    \ENDFOR
    \RETURN{$u_{\mathrm{NN}}, v_{\mathrm{NN}}, p_{\mathrm{NN}}$ according to \cref{equ: u_NN}.}
    \end{algorithmic}
\end{algorithm}

\subsection{Numerical analysis}
\label{sec: met_na}
In this section, we conduct a numerical analysis of the HFS-TransNet method. 
Specifically, under an approximability assumption, we provide an \textit{a priori} error estimate for the physics-informed module (i.e., FS-TransNet) and derive its error upper bound.

The numerical error of the FS-TransNet method can be decomposed into two terms: the network approximation error and the iteration error, as follows:
\begin{equation}
    \left\{
    \begin{aligned}
        \boldsymbol{u}_{\mathrm{NN}}^{n+1} - \boldsymbol{u} &= \left(\boldsymbol{u}_{\mathrm{NN}}^{n+1} - \boldsymbol{u}^{n+1}\right) + \left(\boldsymbol{u}^{n+1} - \boldsymbol{u}\right), \\
        p_{\mathrm{NN}}^{n+1} - p &= \left(p_{\mathrm{NN}}^{n+1} - p^{n+1}\right) + \left(p^{n+1} - p\right). 
    \end{aligned}\label{equ: err_parts}
    \right.
\end{equation}
For the network approximation errors $\left(\boldsymbol{u}_{\mathrm{NN}}^{n+1} - \boldsymbol{u}^{n+1}\right)$ and $\left(p_{\mathrm{NN}}^{n+1} - p^{n+1}\right)$, we assume, based on the universal approximation theorem for single-hidden-layer feedforward networks (SLFNs) \cite{hornik1991approximation}, that the TransNet can approximate the solutions of the two subproblems within certain error bounds. 
For the iteration errors $e_{\boldsymbol{u}}^{n+1}:=\boldsymbol{u}^{n+1} - \boldsymbol{u}$ and $e_{p}^{n+1}:=p^{n+1} - p$, which satisfy homogeneous boundary conditions, we derive an error estimate for $\partial_t e_p^{n+1}$ via the fixed-stress splitting method.

First, we introduce the bound of the network approximation error. From \cref{equ: u_NN}, the space of SLFNs is defined as follows:
\begin{equation}
    V_I^\sigma = \left\{ u_{\mathrm{NN}}: \mathbb{R}^d \to \mathbb{R} \ \bigg|\ u_{\mathrm{NN}}(\boldsymbol{x})=\sum_{i=1}^I \alpha_i \sigma\left(\boldsymbol{w}_i^{\top} \boldsymbol{x}+b_i\right) \right\}.
\end{equation}
The universal approximation property of SLFNs has been established in \cite{hornik1991approximation}.
\begin{theorem}[Universal Approximation \cite{hornik1991approximation}]\label{thm: univeral approximation}
    Suppose $ 1 \leq p < \infty, 0 \leq s < \infty $, and $\Omega \subset \mathbb{R}^d$ is bounded. If $ \sigma \in C^s(\Omega)$ is non-constant and bounded, then the space
    $V^\sigma := \bigcup_{I=1}^\infty V_I^\sigma$
    is dense in the Sobolev space
\begin{equation}
    W^{s,p}(\Omega) := \{ v \in L^p(\Omega) : D^\alpha v \in L^p(\Omega),\ \forall |\alpha| \leq s \}. \nonumber
\end{equation}
\end{theorem}

Assume that $(\boldsymbol{u}, p)$ is the solution of the problem \cref{equ: fluid mass conservation}--\cref{equ: initial condition}. 
Let $(\boldsymbol{u}^{n+1}, p^{n+1})$ denote the solution of fixed-stress splitting method, and $(\boldsymbol{u}_{\mathrm{NN}}^{n+1}, p_{\mathrm{NN}}^{n+1})$ denote the corresponding network approximations obtained from \cref{equ: Step 1 PDE NN}--\cref{equ: Step 2 Initial NN}.
According to \cref{thm: univeral approximation}, for any $\delta > 0$, there exists a sufficiently large $I$ such that
\begin{equation}\label{equ: ass}
    \|\boldsymbol{u}_{\mathrm{NN}}^{n+1} - \boldsymbol{u}^{n+1}\|_{H^1([0,T];\boldsymbol{H}^1(\Omega))} \leq \delta, \quad \|p_{\mathrm{NN}}^{n+1} - p^{n+1}\|_{H^1([0,T];\boldsymbol{H}^1(\Omega))} \leq \delta.
\end{equation}
Then, we establish the error estimate of the iteration \cref{equ: Step 1 PDE NN}--\cref{equ: Step 2 Initial NN}.

\begin{proposition}\label{pro: FS-TransNet err_estimate}
    Under assumption \cref{equ: ass}, if $\beta_{FS}$ is sufficiently large, then there exists a positive number $\tilde{L} < 1$ such that the following estimate holds for all $n \geq 0$,
    \begin{equation}\label{equ: FS-TransNet error contraction}
        \int_0^T \int_{\Omega} \left(\partial_t e_p^{n+1}\right)^2 dxdt \leq \tilde{L} \int_0^T \int_{\Omega} \left(\partial_t e_p^n\right)^2 dxdt + C\delta^2,
    \end{equation}
    where $C$ is a positive constant. Moreover, the limit superior exists and satisfies 
    \begin{equation}\label{equ: FS-TransNet error upper bound}
        \limsup_{n \to \infty} \int_0^T \int_{\Omega} \left(\partial_t e_p^{n+1}\right)^2 dxdt \leq C \frac{\delta^2}{1-\tilde{L}}.
    \end{equation}
\end{proposition}
\cref{pro: FS-TransNet err_estimate} provides an error-propagation estimate of the FS-TransNet method.
For the sake of brevity, the proof is moved to Appendix \ref{app: fsTransNet}.

\begin{remark}
    \cref{pro: FS-TransNet err_estimate} should be interpreted with care.
    First, it is conditional on \cref{equ: ass}, which is an existence statement in the idealized limit $I\to\infty$.
    Second, even granting \cref{equ: ass}, \cref{equ: FS-TransNet error contraction} is a contraction up to an additive $\mathcal{O}(\delta^2)$ term, not a strict contraction to zero.
    Therefore, \cref{pro: FS-TransNet err_estimate} should be understood as a conditional error-propagation estimate for the fixed-stress iteration under inexact (network-approximated) subproblem solves.
\end{remark}
\begin{remark}
    While a rigorous numerical analysis of the data-driven module remains a theoretical challenge and is reserved for future research, its practical effectiveness is well-supported by the numerical experiments in \Cref{sec: ne}.
    Specifically, the ablation study in \cref{fig: case1_ablation_iter} demonstrates that the data-driven initialization achieves significantly higher accuracy than a zero initialization baseline within the same number of iterations. 
    Furthermore, the results in \cref{tab: case1_ablation_init_val} show that the greedy update strategy and the relative tolerance criterion remain robust across varying initialization and validation domains. 
    These results collectively indicate that the data-driven module accelerates convergence without compromising the solution accuracy.
\end{remark}

\section{Numerical experiments}
\label{sec: ne}
In this section, we conduct numerical experiments to evaluate the performance of the proposed HFS-TransNet method on a series of Biot models in poroelasticity.
The implementation details including network parameters for all experiments are listed in Appendix \ref{app: para}.

As in \cite{lu2025multiple}, the discrete relative $\text{L}_2$ error (denoted as $\text{RL}_2$) and the discrete relative $\text{L}_{\infty}$ error (denoted as $\text{RL}_\infty$) are used to evaluate the network solution $u_{\mathrm{NN}}$
\begin{equation}
    \text{RL}_2(u) = \sqrt{\frac{\sum_{j=1}^{J} \left(u_{\mathrm{NN}}(\boldsymbol{x}_j) - u(\boldsymbol{x}_j)\right)^2}{\sum_{j=1}^{J} u^2(\boldsymbol{x}_j)}}, \quad
    \text{RL}_\infty(u) = \frac{\max\left\{|u_{\mathrm{NN}}(\boldsymbol{x}_j) - u(\boldsymbol{x}_j)|\right\}_{j=1}^{J}}{\max\left\{|u(\boldsymbol{x}_j)|\right\}_{j=1}^{J}}, \nonumber
\end{equation}
where $\boldsymbol{x}_j$ is the $j$-th test point and $J$ denotes the total number of test points.

\subsection{Test Case 1}
\label{sec: ne_case1}
We consider the quasi-static Biot model on the space-time domain $\Omega \times [0, T] = [0,1]^2 \times [0, 1]$, with pure Dirichlet boundary conditions applied as in \cite{cai2023combination, feng2018analysis}. The source $f$, body force $\boldsymbol{g}$, initial conditions, and boundary conditions are chosen such that the exact solutions are given as $\boldsymbol{u}(x,y,t) = \left(tx^2/2, ty^2/2\right)^T, p(x,y,t) = \sin(x+y) e^t$.
The model parameters are set to $E=1.0$, $\nu=0.3, 1/M=1.0, \alpha=1.0, K=1.0$.

\subsubsection{Error estimate}
\label{sec: ne_case1_err}

Combining \cref{equ: err_parts}, \cref{equ: ass}, and \cref{equ: FS-TransNet error upper bound} yields 
\begin{equation}
    \Delta := \max_{n\geq0}\|\partial_t p_{\mathrm{NN}}^{n+1} - \partial_t p\|_{L^2} \sim \mathcal{O}(\delta). \nonumber
\end{equation}
To verify this, we define the following errors over $\Lambda = \{1,2,3,4,5\}$ iterations:
\begin{equation}
    \left\{
    \begin{aligned}
    \Delta^\Lambda &:= \max_{n \in \Lambda} \|\partial_t p_{\mathrm{NN}}^{n+1} - \partial_t p\|_{L^2} \\
    \delta_{H^1}^\Lambda &:= \max_{n \in \Lambda} \left\{\|\boldsymbol{u}_{\mathrm{NN}}^{n+1} - \boldsymbol{u}^{n+1}\|_{H^1}, \|p_{\mathrm{NN}}^{n+1} - p^{n+1}\|_{H^1} \right\} \\
    \delta_{L^2}^\Lambda &:= \max_{n \in \Lambda} \left\{\|\boldsymbol{u}_{\mathrm{NN}}^{n+1} - \boldsymbol{u}^{n+1}\|_{L^2}, \|p_{\mathrm{NN}}^{n+1} - p^{n+1}\|_{L^2} \right\} \\
    \end{aligned}\nonumber
    \right.
\end{equation}
\cref{tab: case1_err} reports these quantities for increasing numbers of hidden-layer neurons, where the ``Orders'' column records the log-log slope of $\Delta^\Lambda$ against $\delta_{H^1}^\Lambda$ and $\delta_{L^2}^\Lambda$. 
The observed slopes remaining close to $1$ across all neuron counts confirm the predicted linear scaling $\Delta^\Lambda \sim \mathcal{O}(\delta^\Lambda)$, providing quantitative support for the error-propagation estimate in \cref{pro: FS-TransNet err_estimate}.
\begin{table}[htbp]
    \footnotesize
    \caption{Error estimate in Test Case 1.}
    \label{tab: case1_err}
    \begin{center}
    \begin{tabular}{c c c c}
    \hline
    Neurons & $\delta_{H^1}^\Lambda \, \& \, \delta_{L^2}^\Lambda$ & $\Delta^\Lambda$ & Orders\\
    \hline
    $20$ & 1.86e$-$1 \& 3.59e$-$2 & 1.40e$-$1 & - \\
    $40$ & 5.74e$-$2 \& 1.17e$-$2 & 4.34e$-$2 & 0.99 \& 1.04 \\
    $80$ & 1.47e$-$2 \& 3.33e$-$3 & 9.04e$-$3 & 1.15 \& 1.24 \\
    \hline
    \end{tabular}
    \end{center}
\end{table}

\subsubsection{Ablation for the number of iterations}
\label{sec: ne_case1_iter}
In this part, we compare the accuracy of the two iterative algorithms incorporating the fixed-stress (FS-TransNet and HFS-TransNet) during the first three iterations. 
\cref{fig: case1_ablation_iter} presents the convergence of the errors $\text{RL}_{2}(p)$ and $\text{RL}_{\infty}(p)$ for both methods as the number of hidden-layer neurons increases.
Due to the data-driven initialization strategy, HFS-TransNet achieves significantly higher accuracy during the initial three iterations compared to FS-TransNet.
This highlights HFS-TransNet's capacity to achieve superior accuracy with a substantially reduced computational budget.
\begin{figure}[tbhp]
    \centering
	\includegraphics[width=.9\textwidth]{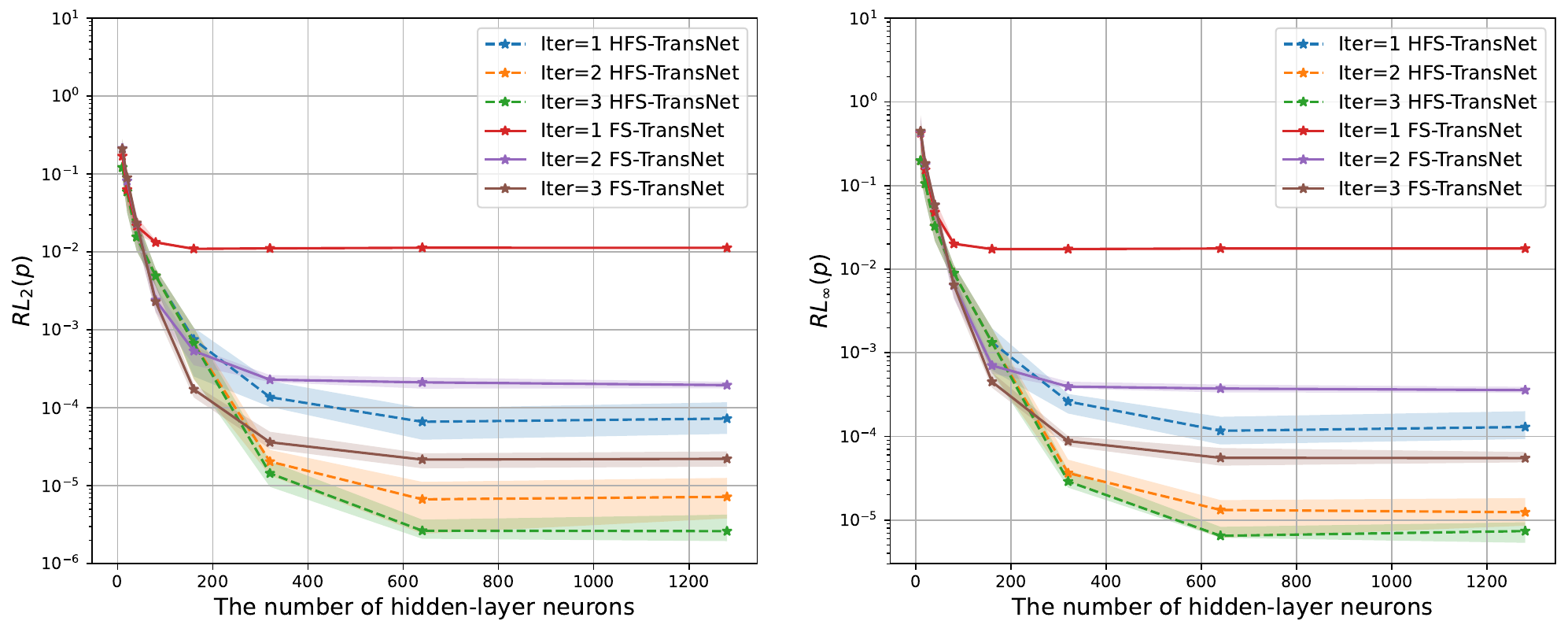}
	\vspace{-0.2cm}
	\caption{Comparison of the $\text{RL}_2(p)$ (left) and $\text{RL}_{\infty}(p)$ (right) errors between FS-TransNet and HFS-TransNet in Test Case 1 across varying iterations. 
    The confidence band represents the range between the 20th and 80th percentiles across 20 repeated runs.}
	\label{fig: case1_ablation_iter}
\end{figure}

\subsubsection{Ablation for the initialization and validation domains}
\label{sec: ne_case1_data}
The initialization and validation domains are two crucial design choices within the HFS-TransNet framework. 
Since the initialization domain determines the initial guess for the fixed-stress splitting iterations, its selection directly dictates the error in subsequent steps, thereby impacting overall computational efficiency. 
Meanwhile, the validation domain influences whether to update iterative weights, ultimately shaping the precision of the output.
\cref{tab: case1_ablation_init_val} summarizes the ablation results, where $\text{T}_\text{init}$ and $\text{T}_\text{val}$ denote the time intervals of the initialization domain and validation domain respectively.
As these domains vary, the resulting errors consistently stabilize within the range of $\mathcal{O}(10^{-5})$ to $\mathcal{O}(10^{-6})$. 
This insensitivity significantly reduces the cost of hyperparameter tuning.
\begin{table}[htbp]
    \footnotesize
    \caption{Ablation study of initialization and validation domains in Test Case 1.}
    \label{tab: case1_ablation_init_val}
    \begin{center}
    \begin{tabular}{llllllll}
    \hline
    $\text{T}_\text{init}$ & $\text{T}_\text{val}$
    & $\text{RL}_2(u)$ & $\text{RL}_\infty(u)$
    & $\text{RL}_2(v)$ & $\text{RL}_\infty(v)$
    & $\text{RL}_2(p)$ & $\text{RL}_\infty(p)$\\
    \hline
        & $ (0.10, 0.20]$ & 1.96e$-$5 & 2.93e$-$5 & 1.19e$-$5 & 2.11e$-$5 & 1.33e$-$5 & 2.01e$-$5\\
    $ [0, 0.10]$ & $ (0.10, 0.25]$ & 1.96e$-$5 & 2.89e$-$5 & 1.57e$-$5 & 2.42e$-$5 & 2.74e$-$5 & 3.94e$-$5\\
        & $ (0.10, 0.30]$ & 1.43e$-$5 & 2.27e$-$5 & 1.71e$-$5 & 2.36e$-$5 & 2.44e$-$6 & 6.86e$-$6\\ 
    \hline
        & $ (0.15, 0.25]$ & 1.16e$-$5 & 2.14e$-$5 & 1.10e$-$5 & 1.63e$-$5 & 2.17e$-$6 & 5.20e$-$6\\
    $ [0, 0.15]$ & $ (0.15, 0.30]$ & 1.61e$-$5 & 2.03e$-$5 & 2.29e$-$5 & 3.03e$-$5 & 2.69e$-$6 & 7.19e$-$6\\
        & $ (0.15, 0.35]$ & 1.56e$-$5 & 2.53e$-$5 & 1.53e$-$5 & 2.35e$-$5 & 2.45e$-$6 & 5.63e$-$6\\ 
    \hline
        & $ (0.20, 0.30]$ & 1.51e$-$5 & 2.63e$-$5 & 1.36e$-$5 & 2.28e$-$5 & 3.47e$-$6 & 7.95e$-$6\\
    $ [0, 0.20]$ & $ (0.20, 0.35]$ & 1.95e$-$5 & 2.83e$-$5 & 1.37e$-$5 & 2.08e$-$5 & 4.32e$-$6 & 9.59e$-$6\\
        & $ (0.20, 0.40]$ & 2.89e$-$5 & 3.05e$-$5 & 1.36e$-$5 & 2.28e$-$5 & 2.48e$-$6 & 7.23e$-$6\\ 
    \hline
    \end{tabular}
    \end{center}
\end{table}

\subsubsection{Comparison with FS-FEM and FS-PINN}
\label{sec: ne_case1_vs}
In this part, we compare HFS-TransNet with FS-FEM and FS-PINN.
For FS-FEM, a backward Euler scheme following \cite{ju2020parameter} is utilized for time discretization. The total number of spatial degrees of freedom is set to $\text{DOFs}=2482$, with a time step size of $\Delta t=0.01$.
For FS-PINN, nearly all hyperparameters follow those in \cite{cai2023combination}. For instance, two neural networks are employed to approximate the displacement $\boldsymbol{u}$ and pressure $p$ respectively, each consisting of 3 hidden layers with 15 neurons per layer. Additionally, the learning rate is set to $1\times10^{-4}$ with 50 epochs. The only modification is the fixed-stress iteration count, set to 10 for consistency across all three methods.
For HFS-TransNet, the parameter settings are detailed in \cref{tab: net_para_case1_case2}.

\cref{tab: case1_table} reports the performance comparisons of FS-FEM, FS-PINN, and HFS-TransNet for different values of Poisson's ratio $\nu$ and hydraulic conductivity $K$. 
First, we focus on computational efficiency. Under the same GPU acceleration, HFS-TransNet requires significantly less computational time (including both training and inference) compared to FS-PINN. 
Moreover, compared to the FS-FEM ($\text{DOFs}=2482, \Delta t=0.01$) implemented on a CPU, HFS-TransNet utilizes fewer degrees of freedom ($\text{DOFs}=1920$) and avoids the computational cost of temporal marching.
Second, we analyze the computational accuracy. As $\nu \rightarrow 0.5$ or $K \rightarrow 0$, the precision of all three methods undergoes a decline. Overall, both FS-FEM and HFS-TransNet significantly outperform FS-PINN. Specifically, FS-FEM exhibits superior performance in the displacement field, whereas HFS-TransNet achieves higher accuracy in the pressure field.
\begin{table}[htbp]
    \footnotesize
    \caption{Comparison between FS-FEM, FS-PINN, and HFS-TransNet in Test Case 1.}
    \label{tab: case1_table}
    \setlength{\tabcolsep}{3.5pt}
    \begin{center}
    \begin{tabular}{lllcllllll}
    \hline
    Method & $\nu$ & $K$ & Time(s)
    & $\text{RL}_2(u)$ & $\text{RL}_\infty(u)$
    & $\text{RL}_2(v)$ & $\text{RL}_\infty(v)$
    & $\text{RL}_2(p)$ & $\text{RL}_\infty(p)$\\
    \hline
    \multirow{4}*{FS-FEM} 
    & 0.300 & 1.00 & 42.9 & 1.10e$-$4 & 1.20e$-$4 & 1.10e$-$4 & 1.15e$-$4 & 2.00e$-$4 & 3.03e$-$4\\
    & 0.499 & 1.00 & 42.2 & 2.01e$-$6 & 3.25e$-$6 & 2.05e$-$6 & 3.80e$-$6 & 2.02e$-$4 & 3.07e$-$4\\ 
    & 0.300 & 0.01 & 41.8 & 3.81e$-$4 & 3.42e$-$4 & 3.85e$-$4 & 3.43e$-$4 & 1.60e$-$3 & 1.82e$-$3\\
    & 0.499 & 0.01 & 41.9 & 5.58e$-$6 & 5.40e$-$6 & 5.75e$-$6 & 5.55e$-$6 & 1.72e$-$3 & 2.06e$-$3\\
    \hline
    \multirow{4}*{FS-PINN} 
    & 0.300 & 1.00 & 597 & 1.33e$-$2 & 1.34e$-$2 & 1.09e$-$2 & 1.54e$-$2 & 4.00e$-$3 & 7.09e$-$3\\       
    & 0.499 & 1.00 & 647 & 1.40e$-$0 & 1.46e$-$0 & 1.56e$+$0 & 1.53e$+$0 & 5.00e$-$2 & 6.87e$-$2\\
    & 0.300 & 0.01 & 636 & 2.37e$-$2 & 4.00e$-$2 & 2.09e$-$2 & 2.88e$-$2 & 4.18e$-$3 & 6.16e$-$3\\       
    & 0.499 & 0.01 & 611 & 1.09e$-$0 & 1.41e$-$0 & 7.95e$-$1 & 9.20e$-$1 & 3.94e$-$1 & 4.30e$-$1\\
    \hline
    \multirow{4}*{HFS-TransNet} 
    & 0.300 & 1.00 & 1.67 & 1.64e$-$5 & 2.36e$-$5 & 1.01e$-$5 & 1.41e$-$5 & 2.03e$-$6 & 5.20e$-$6\\
    & 0.499 & 1.00 & 1.65 & 1.34e$-$5 & 2.02e$-$5 & 1.08e$-$5 & 1.77e$-$5 & 2.77e$-$6 & 5.56e$-$6\\
    & 0.300 & 0.01 & 0.90 & 2.16e$-$4 & 2.70e$-$4 & 3.89e$-$4 & 4.41e$-$4 & 8.51e$-$4 & 9.97e$-$4\\
    & 0.499 & 0.01 & 2.33 & 1.06e$-$5 & 1.51e$-$5 & 1.60e$-$5 & 2.58e$-$5 & 2.53e$-$6 & 5.53e$-$6\\
    \hline
    \end{tabular}
    \end{center}
\end{table}

\subsection{Test Case 2}
\label{sec: ne_case2}
We consider the Biot model on the space-time domain $\Omega \times [0, T] = [0,1]^2 \times [0, 1]$, with mixed boundary conditions applied as in \cite{cai2023combination, ju2020parameter, yi2017study}. 
In this experiment, the source $f$, body force $\boldsymbol{g}$, and initial conditions are chosen such that the exact solutions are given as follows: 
\begin{equation}
    \boldsymbol{u}=e^{-t}
    \begin{pmatrix}
    \sin(2\pi y)(\cos(2\pi x)-1)+\frac{1}{\mu + \lambda}\sin(\pi x)\sin(\pi y) \\
    \sin(2\pi x)(1 - \cos(2\pi y))+\frac{1}{\mu + \lambda}\sin(\pi x)\sin(\pi y)
    \end{pmatrix}, \quad
    p = e^{-t}\sin(\pi x)\sin(\pi y). \nonumber
\end{equation}
The boundaries are denoted as $\Gamma_1=\{(1, y) \mid 0\leq y \leq 1\}$, $\Gamma_2=\{(x, 0) \mid 0\leq x \leq 1\}$, $\Gamma_3=\{(0, y) \mid 0\leq y \leq 1\}$, $\Gamma_4=\{(x, 1) \mid 0\leq x \leq 1\}$, and the mixed boundary conditions are given as follows:
\begin{equation}
    \left\{
    \begin{aligned}
        &p = e^{-t}\sin(\pi x)\sin(\pi y) &\text{on } (\Gamma_1 \cup \Gamma_3) \times (0, T], \\
        &u_1 = e^{-t}\left( \sin(2\pi y)(\cos(2\pi x) - 1) + \frac{1}{\mu + \lambda} \sin(\pi x)\sin(\pi y) \right) &\text{on } (\Gamma_1 \cup \Gamma_3) \times (0, T], \\
        &u_2 = e^{-t}\left( \sin(2\pi x)(1 - \cos(2\pi y)) + \frac{1}{\mu + \lambda} \sin(\pi x)\sin(\pi y) \right) &\text{on } (\Gamma_1 \cup \Gamma_3) \times (0, T], \\
        &\sigma(\boldsymbol{u}) \boldsymbol{n} - \alpha p \boldsymbol{n} = \boldsymbol{h} &\text{on } (\Gamma_2 \cup \Gamma_4) \times (0, T], \\
        &K\nabla p \cdot \boldsymbol{n} = Ke^{-t}\left( \pi \cos(\pi x)\sin(\pi y)n_1 + \pi \sin(\pi x)\cos(\pi y)n_2 \right) &\text{on } (\Gamma_2 \cup \Gamma_4) \times (0, T], 
    \end{aligned} \nonumber
    \right.
\end{equation}
where the normal vector is denoted as $\boldsymbol{n}=(n_1, n_2)^\top$ and $\boldsymbol{h}$ is chosen such that the exact solution is satisfied. 
The model parameters in Test Case 2 are set to $E=1.0, \nu=0.3, 1/M=1.0, \alpha=1.0, K=1.0$.

\subsubsection{Ablation for the Poisson's ratio}
\label{sec: ne_case2_nu}
In this part, we test the robustness of three TransNet-based algorithms (TransNet, FS-TransNet, and HFS-TransNet) with respect to the Poisson's ratio $\nu$. 
The values of $\nu$ are set to $0.3, 0.5-10^{-3}$, and $0.5-10^{-6}$, and the results are reported in \cref{tab: case2_ablation_nu}..
For $\nu=0.3$ and $0.5-10^{-3}$, all three methods exhibit high accuracy in resolving the coupled displacement-pressure fields. 
However, as $\nu$ approaches the incompressible limit ($\nu=0.5-10^{-6}$), the accuracy of both TransNet and FS-TransNet deteriorates significantly. 
While TransNet loses accuracy in both fields and FS-TransNet struggles with the displacement component despite a reasonable pressure estimate, HFS-TransNet retains its accuracy for the coupled solutions.
While FS-TransNet requires a higher computational cost (including both training and inference) than conventional TransNet due to the fixed-stress splitting iteration, HFS-TransNet achieves a shorter runtime than FS-TransNet by incorporating the data-driven module.
\begin{table}[htbp]
    \footnotesize
    \caption{Ablation study of parameter $\nu$ in Test Case 2.}
    \label{tab: case2_ablation_nu}
    \setlength{\tabcolsep}{3.5pt}
    \begin{center}
    \begin{tabular}{lrclllllll}
    \hline
    $\nu$ & Method & Time(s)
    & $\text{RL}_2(u)$ & $\text{RL}_\infty(u)$
    & $\text{RL}_2(v)$ & $\text{RL}_\infty(v)$
    & $\text{RL}_2(p)$ & $\text{RL}_\infty(p)$\\
    \hline
            & TransNet     & 7.95 & 1.10e$-$4 & 1.41e$-$4 & 7.01e$-$5 & 9.44e$-$5 & 2.81e$-$5 & 8.04e$-$5\\
    $0.3$   & FS-TransNet  & 83.3 & 3.76e$-$4 & 3.35e$-$4 & 1.49e$-$4 & 2.56e$-$4 & 2.39e$-$4 & 2.92e$-$4\\
            & HFS-TransNet & 17.5 & 1.27e$-$3 & 2.23e$-$3 & 1.17e$-$3 & 1.44e$-$3 & 6.75e$-$3 & 4.73e$-$3\\
    \hline
                    & TransNet     & 7.88 & 1.95e$-$3 & 4.36e$-$3 & 1.90e$-$3 & 3.90e$-$3 & 2.51e$-$4 & 5.63e$-$4\\ 
    $0.5-10^{-3}$   & FS-TransNet  & 84.6 & 2.92e$-$3 & 5.04e$-$3 & 2.23e$-$3 & 5.94e$-$3 & 2.43e$-$5 & 3.90e$-$5\\
                    & HFS-TransNet & 26.3 & 1.37e$-$3 & 2.28e$-$3 & 1.47e$-$3 & 1.94e$-$3 & 2.00e$-$5 & 2.61e$-$5\\
    \hline
                    & TransNet     & 7.84 & 7.44e$-$1 & 1.50e$-$0 & 7.96e$-$1 & 1.09e$-$0 & 5.69e$-$2 & 1.89e$-$1\\
    $0.5-10^{-6}$   & FS-TransNet  & 85.5 & 4.10e$-$1 & 5.76e$-$1 & 4.79e$-$1 & 4.15e$-$1 & 2.74e$-$5 & 3.38e$-$5\\
                    & HFS-TransNet & 8.65 & 2.35e$-$2 & 1.89e$-$2 & 2.48e$-$2 & 3.00e$-$2 & 6.75e$-$3 & 4.73e$-$3\\
    \hline
    \end{tabular}
    \end{center}
\end{table}

\subsubsection{Ablation for the hydraulic conductivity}
\label{sec: ne_case2_k}
To evaluate the impact of hydraulic conductivity, we vary $K$ among $1.0$, $10^{-3}$, and $10^{-6}$, with the results reported in \cref{tab: case2_ablation_K}.
Across the range of $K$ values, these three TransNet-based methods consistently resolve the coupled displacement and pressure fields.
\begin{table}[htbp]
    \footnotesize
    \caption{Ablation study of parameter $K$ in Test Case 2.}
    \label{tab: case2_ablation_K}
    \setlength{\tabcolsep}{3.5pt}
    \begin{center}
    \begin{tabular}{lrclllllll}
    \hline
    $K$  & Method & Time(s)
    & $\text{RL}_2(u)$ & $\text{RL}_\infty(u)$
    & $\text{RL}_2(v)$ & $\text{RL}_\infty(v)$
    & $\text{RL}_2(p)$ & $\text{RL}_\infty(p)$\\
    \hline
            & TransNet     & 7.95 & 1.10e$-$4 & 1.41e$-$4 & 7.01e$-$5 & 9.44e$-$5 & 2.81e$-$5 & 8.04e$-$5\\
    $1.0$   & FS-TransNet  & 83.3 & 3.76e$-$4 & 3.35e$-$4 & 1.49e$-$4 & 2.56e$-$4 & 2.39e$-$4 & 2.92e$-$4\\
            & HFS-TransNet & 17.5 & 1.27e$-$3 & 2.23e$-$3 & 1.17e$-$3 & 1.44e$-$3 & 6.75e$-$3 & 4.73e$-$3\\
    \hline
                & TransNet     & 7.87 & 7.66e$-$5 & 1.03e$-$4 & 5.93e$-$5 & 7.96e$-$5 & 1.21e$-$4 & 3.62e$-$4\\ 
    $10^{-3}$   & FS-TransNet  & 89.0 & 2.88e$-$4 & 4.86e$-$4 & 1.53e$-$4 & 2.90e$-$4 & 7.18e$-$4 & 3.00e$-$3\\
                & HFS-TransNet & 17.4 & 1.27e$-$3 & 2.23e$-$3 & 1.17e$-$3 & 1.44e$-$3 & 6.75e$-$3 & 4.73e$-$3\\
    \hline
                & TransNet     & 8.10 & 9.03e$-$5 & 1.12e$-$4 & 6.71e$-$5 & 7.20e$-$5 & 1.34e$-$4 & 3.21e$-$4\\
    $10^{-6}$   & FS-TransNet  & 84.4 & 4.93e$-$4 & 3.83e$-$4 & 2.14e$-$4 & 2.54e$-$4 & 7.53e$-$4 & 3.26e$-$3\\
                & HFS-TransNet & 18.2 & 1.27e$-$3 & 2.23e$-$3 & 1.17e$-$3 & 1.44e$-$3 & 6.75e$-$3 & 4.73e$-$3\\
    \hline
    \end{tabular}
    \end{center}
\end{table}

\subsection{Barry--Mercer problem}
\label{sec: ne_bm}
The Barry--Mercer problem is a classical benchmark problem of porous media involving fluid injection \cite{barry1999exact, cai2023combination, haghighat2022physics, phillips2009overcoming}. It is defined on a two-dimensional square domain $\Omega=[0,1]^2$ with a point source $f$ located at $(x_0, y_0)=(0.25, 0.25)$. In this experiment, the source $f$ and body force $\boldsymbol{g}$ are defined as $f = 2\delta(x-x_0)\delta(y-y_0)sin(t)$ and $\boldsymbol{g}=\boldsymbol{0}$, where $\delta$ denotes the Dirac function. Following \cite{cai2023combination, haghighat2022physics}, a Gaussian distribution is used to approximate the Dirac function $\delta(x) \approx e^{-(x/c_g)^2}/c_g\sqrt{\pi}$ with $c_g = 0.04$. 
\cref{fig: BarryMercer_Boundary} depicts the boundary conditions, and the initial conditions are given as $\boldsymbol{u}=\boldsymbol{0}$, $p=0$ in $\Omega \times \{0\}$.
The model parameters for the Barry--Mercer problem are set to $E=14/15, \nu=1/6, 1/M=0, \alpha=1.0, K=1.0$.
A conventional FS-FEM method with $\Delta t=\pi/200$ is used to generate the high-quality reference solution $u_\mathrm{ref}$.

\cref{fig: BarryMercer_result} compares the displacement and pressure fields $(\boldsymbol{u}, p)$ obtained by FS-FEM and the proposed HFS-TransNet at $t=\pi/2$. 
Furthermore, \cref{fig: BarryMercer_error} presents the pointwise absolute errors between these two solutions.
It is evident that HFS-TransNet effectively captures the complex coupled physics in the vicinity of the source point in the Barry--Mercer problem.
\begin{figure}[tbhp]
    \centering
    \vspace{-0.2cm}
    \subfloat[Boundary conditions]{\label{fig: BarryMercer_Boundary}
    \includegraphics[width=.33\textwidth]{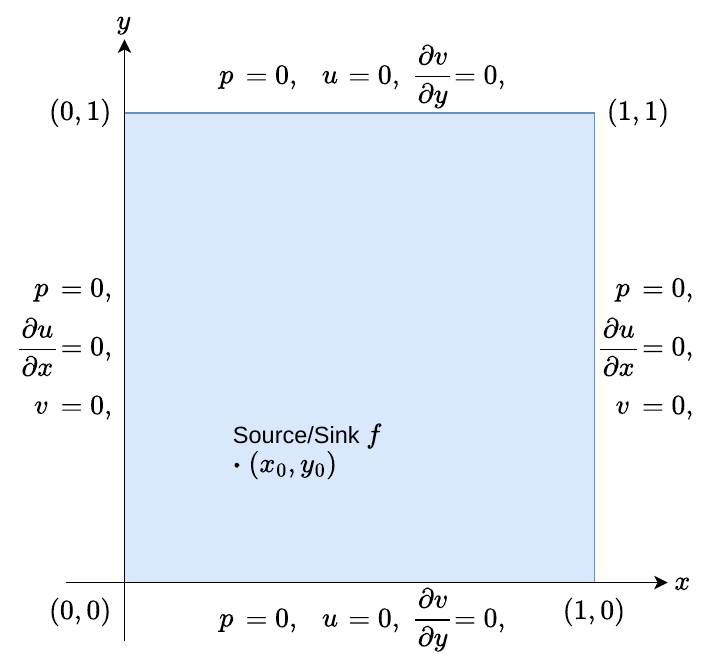}}
    \subfloat[FS-FEM (left) and HFS-TransNet (right) solutions]{\label{fig: BarryMercer_result}
    \includegraphics[width=.65\textwidth]{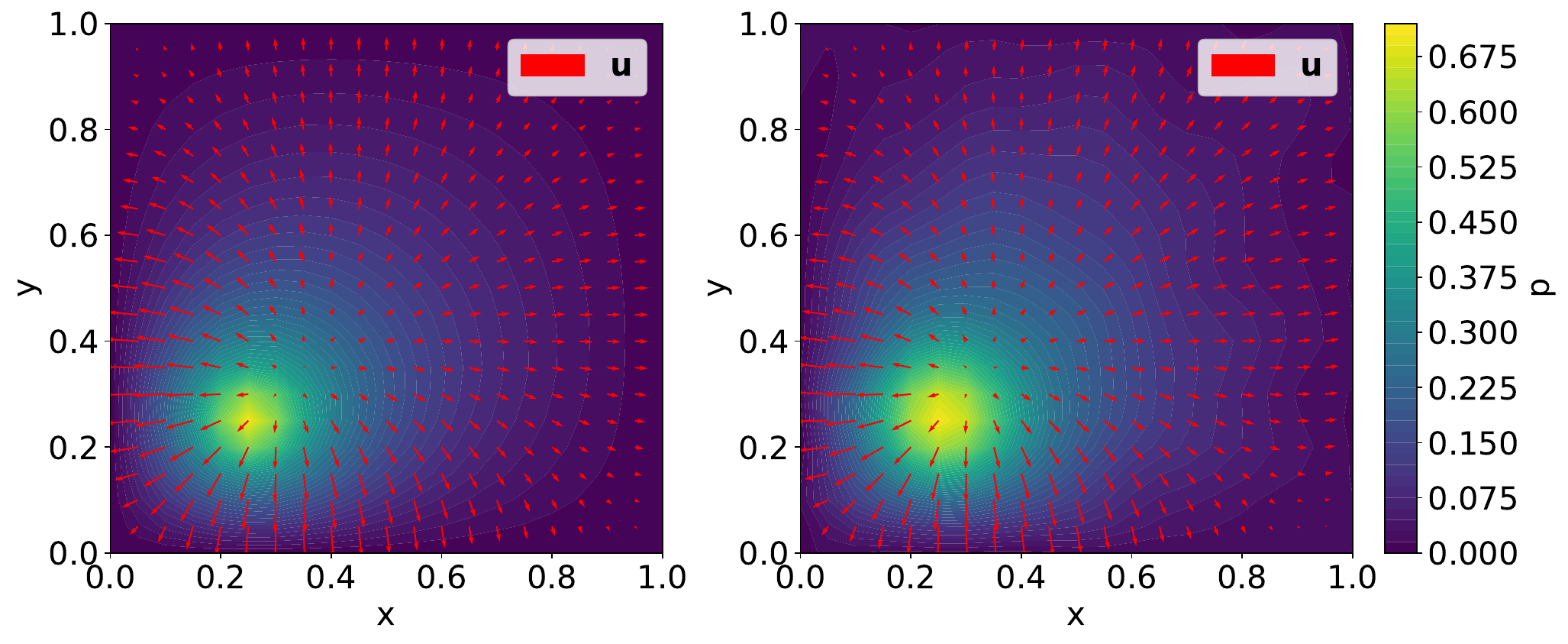}}
	\vspace{-0.2cm}
	\caption{Barry--Mercer problem.}
	\label{fig: BarryMercer}
\end{figure}
\begin{figure}[tbhp]
    \centering
	\includegraphics[width=.98\textwidth]{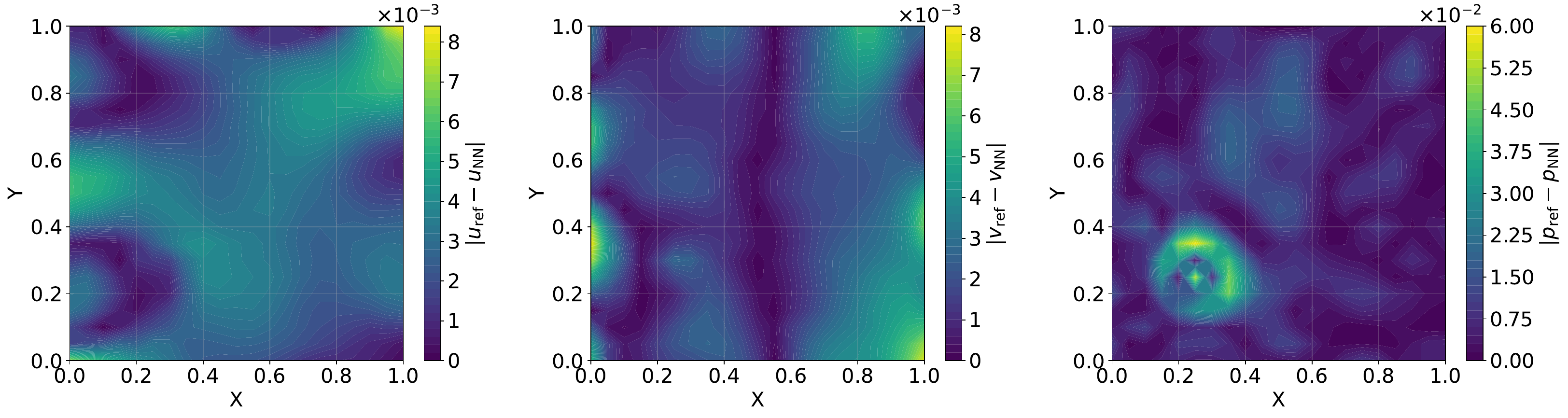}
	\vspace{-0.4cm}
	\caption{Pointwise absolute error for $u$ (left), $v$ (middle), and $p$ (right) between the FS-FEM and HFS-TransNet solutions for the Barry--Mercer problem at $t=\pi/2$.}
	\label{fig: BarryMercer_error}
\end{figure}

\subsection{Mandel problem}
\label{sec: ne_md}
The Mandel problem is a classical benchmark problem known for the Mandel--Cryer effect \cite{abousleiman1996mandel, mandel1953consolidation, phillips2007coupling}. 
This problem considers a saturated poroelastic plate in the domain $\Omega=[-a, a]\times[-b, b]$, constrained between two rigid plates. As illustrated on the left side of \cref{fig: Mandel_Boundary}, the top plate exerts a constant downward force of $2F$, while the bottom plate applies an equal upward force. 
Due to the quarter symmetry, we simplify the computational domain to $[0,a]\times[0,b]$, with the boundary conditions as specified on the right side of \cref{fig: Mandel_Boundary}. 
The source and the body force are set to $f=0$, $\boldsymbol{g}=\boldsymbol{0}$ respectively, while the boundary condition $V$ and the initial conditions follow the analytical series solutions provided in \cite{abousleiman1996mandel, phillips2007coupling}, 
\begin{align}
    \begin{split}
    u=&\left[\frac{F\nu}{2\mu a}-\frac{F\nu_u}{\mu a}\sum_{i = 1}^{\infty}\frac{\sin\beta_i\cos\beta_i}{\beta_i - \sin\beta_i\cos\beta_i}\exp\left(-\beta_i^2 ct/a^2 \right)\right]x \\ 
    &+\frac{F}{\mu}\sum_{i = 1}^{\infty}\frac{\cos\beta_i}{\beta_i - \sin\beta_i\cos\beta_i}\sin\frac{\beta_i x}{a}\exp\left(-\beta_i^2 ct/a^2\right), 
    \end{split}\label{equ: Mandel_exact_u} \\
    v=&\left[-\frac{F(1 - \nu)}{2\mu a}+\frac{F(1 - \nu_u)}{\mu a}\sum_{i = 1}^{\infty}\frac{\sin\beta_i\cos\beta_i}{\beta_i - \sin\beta_i\cos\beta_i}\exp\left(-\beta_i^2 ct/a^2 \right)\right]y, \label{equ: Mandel_exact_v} \\
    p=&\frac{2FB(1 + \nu_u)}{3a} \sum_{i = 1}^{\infty}\frac{\sin\beta_i}{\beta_i - \sin\beta_i\cos\beta_i}\left(\cos\frac{\beta_i x}{a}-\cos\beta_i\right)\exp\left(-\beta_i^2 ct/a^2 \right). \label{equ: Mandel_exact_p}
\end{align}
In the above equations, $\beta_i$ satisfies $\tan\beta_i = (1-\nu)\beta_i/(\nu_u-\nu), i=1, 2, \dots, \infty$, 
and the other parameters for the Mandel problem are listed in \cref{tab: symbols Mandel}.
For the following analysis, the analytical solutions in \cref{equ: Mandel_exact_u}--\cref{equ: Mandel_exact_p} are truncated to 30 terms to serve as the reference. 
\begin{figure}[tbhp]
    \centering
	\includegraphics[width=.95\textwidth]{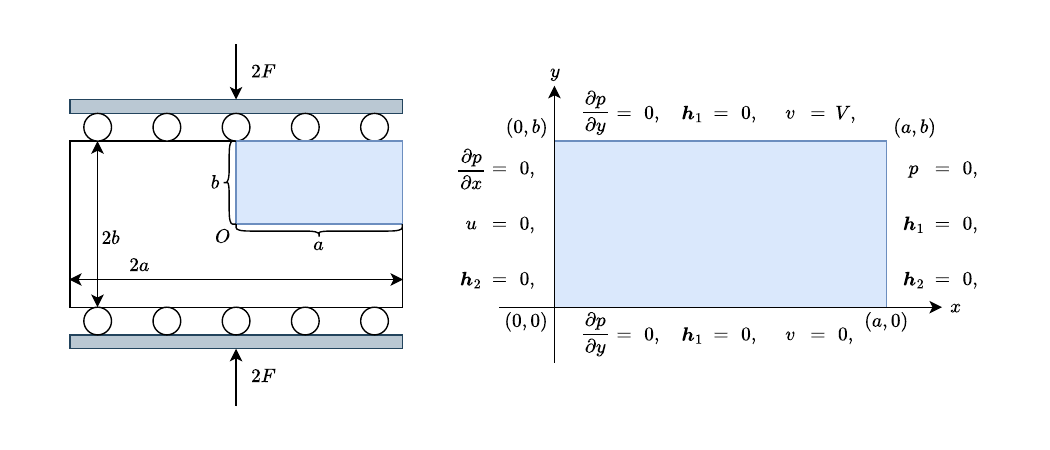}
	\vspace{-0.6cm}
	\caption{Physical description of the Mandel problem (left), and the corresponding computational domain and boundary conditions (right).}
	\label{fig: Mandel_Boundary}
\end{figure}
\begin{table}[htbp]
    \footnotesize
    \caption{Physical parameters for the Mandel problem.}
    \label{tab: symbols Mandel}
    \begin{center}  
    \begin{tabular}{lll}
    \hline
    Symbol & Physical meaning & Value \\
    \hline
    $a$      & Dimension in $x$           & $1.00 \, \text{m}$ \\
    $b$      & Dimension in $y$           & $1.00 \, \text{m}$ \\
    $E$      & Young's modulus            & $5.94 \times 10^9 \, \text{Pa}$ \\
    $\nu$    & Poisson's ratio            & $0.20$ \\
    $\alpha$ & Biot-Willis constant       & $1.00$ \\
    $M$      & Biot modulus               & $1.65 \times 10^{10} \, \text{Pa}$ \\
    $K$      & Hydraulic conductivity     & $9.86 \times 10^{-11} \, \text{m}^2 \cdot \text{Pa}^{-1} \cdot \text{s}^{-1}$ \\
    $F$      & Applied load               & $6.00 \times 10^8 \, \text{N} \cdot \text{m}^{-1}$ \\
    $B$      & Skempton's coefficient     & $0.83$ \\
    $\nu_u$  & Undrained Poisson's ratio  & $0.44$ \\
    $c$      & Diffusivity coefficient    & $0.46 \, \text{m}^2 \cdot \text{s}^{-1}$ \\
    \hline
    \end{tabular}
    \end{center}
\end{table}

Notably, since the pressure $p$ satisfies the governing equation \cite{phillips2007coupling},
\begin{equation}
    \frac{\partial p}{\partial t}-c\frac{\partial^2 p}{\partial x^2}=\frac{2cFB(1 + \nu_u)}{3a^3}\sum_{i = 1}^{\infty}\frac{\beta_i^2 \sin\beta_i\cos\beta_i}{\beta_i - \sin\beta_i\cos\beta_i}\exp(-\beta_i^2 ct/a^2), \nonumber
\end{equation}
the initial $p$ can be obtained via the conventional TransNet method. Subsequently, the initial values of $u$ and $v$ are determined through the physics-informed TransNet, constrained by \cref{equ: fluid mass conservation}--\cref{equ: solid force equilibrium}.
Additionally, the validation stage is bypassed, thereby mitigating the reliance of the data-driven module on observed data.

\cref{fig: Mandel_result} compares the reference and HFS-TransNet solutions at $t=0.001, 0.01, 0.1$, and $1$. 
At each time, results are plotted in the same color, with solid lines denoting the reference and circles representing the HFS-TransNet results.
These results demonstrate that HFS-TransNet accurately captures the temporal evolution of both pressure and displacement fields, showing excellent agreement with the Mandel--Cryer effect.
\begin{figure}[tbhp]
    \centering
	\includegraphics[width=.95\textwidth]{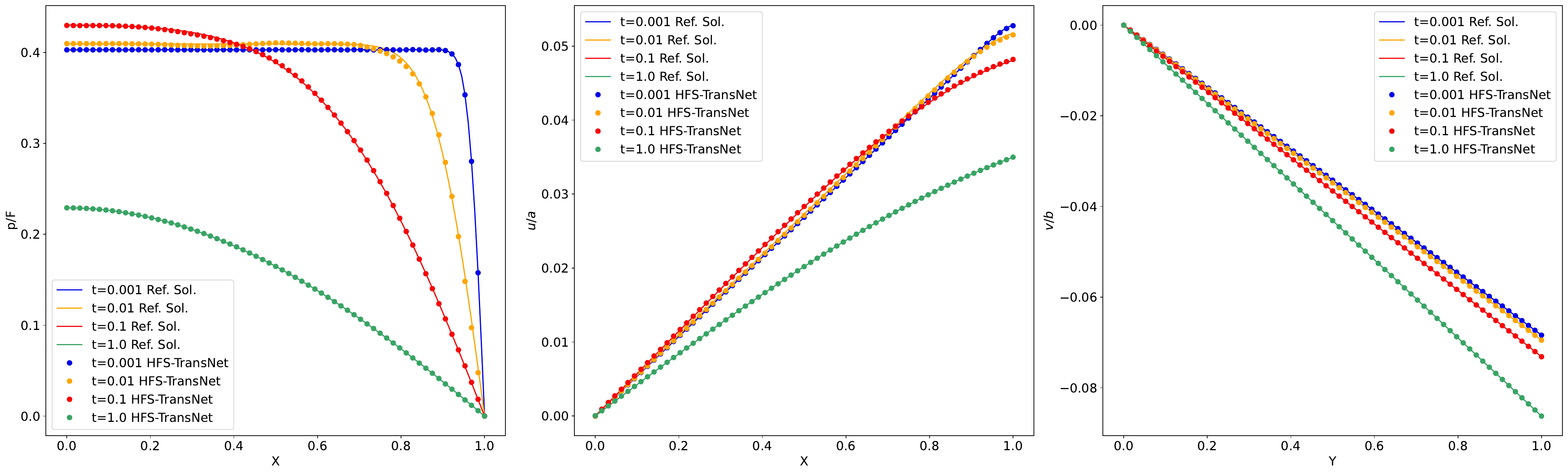}
	\vspace{-0.4cm}
	\caption{Comparison of the pressure $p$ at $y=b$ (left), x-displacement $u$ at $y=b$ (middle), and y-displacement $v$ at $x=a$ (right) obtained by the reference and HFS-TransNet solutions for the Mandel problem.}
	\label{fig: Mandel_result}
\end{figure}

\subsection{Brain edema simulations}
\label{sec: ne_be}
In this section, we explore the application of HFS-TransNet in simulating brain swelling induced by traumatic brain injury (TBI) \cite{chen2023physical, ju2020parameter, wang2025mixed}.
Following \cite{ju2020parameter}, we consider a 2D cross-section of the human brain, as shown in \cref{fig: Brain}.
\begin{figure}[tbhp]
    \centering
    \vspace{-0.2cm}
    \subfloat[A 2D MRI human brain cross-section]{\label{fig: Brain}
    \includegraphics[width=.68\textwidth]{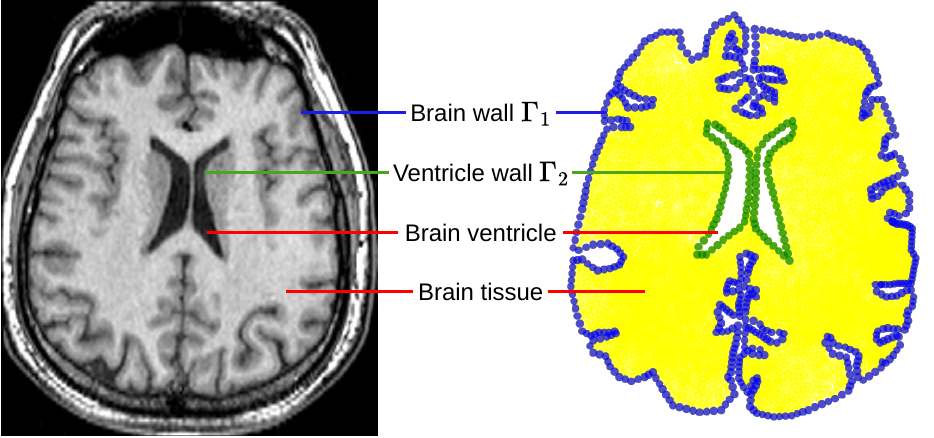}}
    \subfloat[Mesh]{\label{fig: Brain_grid}
    \includegraphics[width=.30\textwidth]{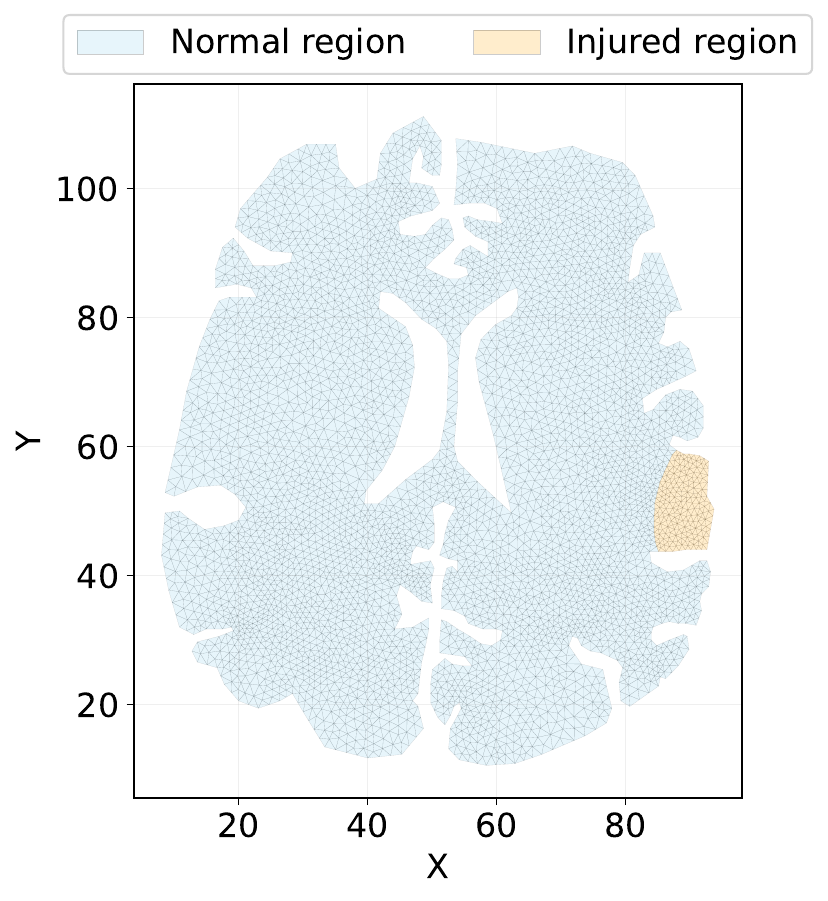}}
	\vspace{-0.2cm}
	\caption{Brain edema simulations.}
	\label{fig: Brain edema simulations}
\end{figure}

In the simulation, body forces are neglected, i.e., $\boldsymbol{g}=\boldsymbol{0}$, and the boundary conditions are prescribed as follows:
\begin{equation}
    \left\{
    \begin{aligned}
        &\boldsymbol{u} = \boldsymbol{0} &\text{on } \Gamma_1 \times (0, T], \\
        &K\nabla p \cdot \boldsymbol{n} = c_b(p_{SAS}-p) &\text{on } \Gamma_1 \times (0, T], \\
        &\sigma(\boldsymbol{u}) \boldsymbol{n} - \alpha p \boldsymbol{n} = -p\boldsymbol{n} &\text{on } \Gamma_2 \times (0, T], \\
        &p=1100\, \text{Pa} &\text{on } \Gamma_2 \times (0, T], 
    \end{aligned} \nonumber
    \right.
\end{equation}
where $c_b$ is the conductance value and $p_{SAS}$ is the pressure in the subarachnoid space (SAS). Furthermore, the parameters in \cref{equ: fluid mass conservation} are defined as $1/M=c_0$ and $K=\kappa/\mu_f$, with their specific values provided in \cref{tab: symbols Brain edema}.
\begin{table}[htbp]
    \footnotesize
    \caption{Physical parameters for brain edema simulations.}
    \label{tab: symbols Brain edema}
    \begin{center}
    \begin{tabular}{lll}
    \hline
    Symbol & Physical meaning & Value \\
    \hline
    $E$      & Young's modulus            & $9010 \, \text{Pa}$ \\
    $\nu$    & Poisson's ratio            & $0.35$ \\
    $\alpha$ & Biot-Willis constant       & $1.00$ \\
    $c_0$    & Specific storage term      & $4.50 \times 10^{-7} \, \text{Pa}^{-1}$ \\
    $\kappa$ & Permeability of the brain  & $1.40 \times 10^{-9} \, \text{mm}^2$ \\
    $\mu_f$  & Fluid viscosity            & $1.48 \times 10^{-5} \, \text{Pa} \cdot \text{min}$ \\
    $c_b$    & Conductance                & $3.00 \times 10^{-5} \, \text{mm} \cdot \text{min}^{-1} \cdot \text{Pa}^{-1}$ \\
    $p_{\text{SAS}}$ & Pressure of subarachnoid space & $1070 \, \text{Pa}$ \\
    \hline
    \end{tabular}
    \end{center}
\end{table}

The reference solution is generated using the FS-FEM on the mesh shown in \cref{fig: Brain_grid} to evaluate the effectiveness of HFS-TransNet. 
The FS-FEM involves $\text{DOFs} = 23,814$ and $\Delta t = 0.01\,\text{min}$, with the fixed-stress iteration count set to $10$.
In the normal region, a state of homeostasis exists between the absorption and discharge of cerebrospinal fluid (CSF), such that $f=0$. Post-TBI, this dynamic equilibrium is disrupted. Therefore, the absorption rate in the injured region is set to $f = 9 \times 10^{-3}\,\text{min}^{-1}$, consistent with \cite{ju2020parameter}.
For HFS-TransNet, following \cite{lu2025multiple}, two separate networks are used to represent the normal $\Omega^n$ and injured $\Omega^i$ regions respectively. Furthermore, a new interface error term is introduced to enforce the continuity conditions across the regional boundaries.

\cref{fig: Brain_HFSTransNet} compare the pressure $p$ and displacement $\boldsymbol{u}$ distributions at $t = 1.5\,\text{min}$, obtained via FS-FEM ($\text{DOFs} = 23814$ with $\Delta t=0.01\,\text{min}$, runtime $1440.9\,\text{s}$ on CPU) and HFS-TransNet ($\text{DOFs} = 12288$, runtime $14.3\,\text{s}$ on GPU) respectively.
Due to the presence of near-zero displacement values, only the absolute error is reported for the displacement field.
The results demonstrate that the proposed HFS-TransNet maintains the precision of conventional FS-FEM while using fewer DOFs and avoiding temporal marching.
\begin{figure}[tbhp]
    \centering
    \includegraphics[width=.95\textwidth]{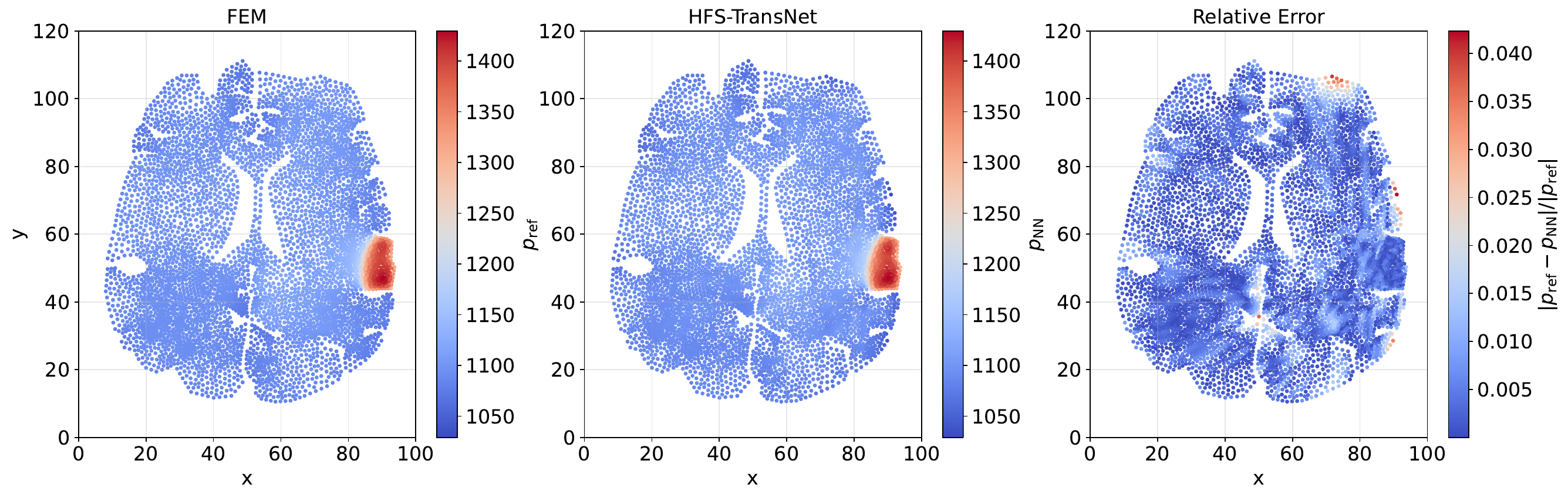}
    \vspace{0.2cm}
    \includegraphics[width=.95\textwidth]{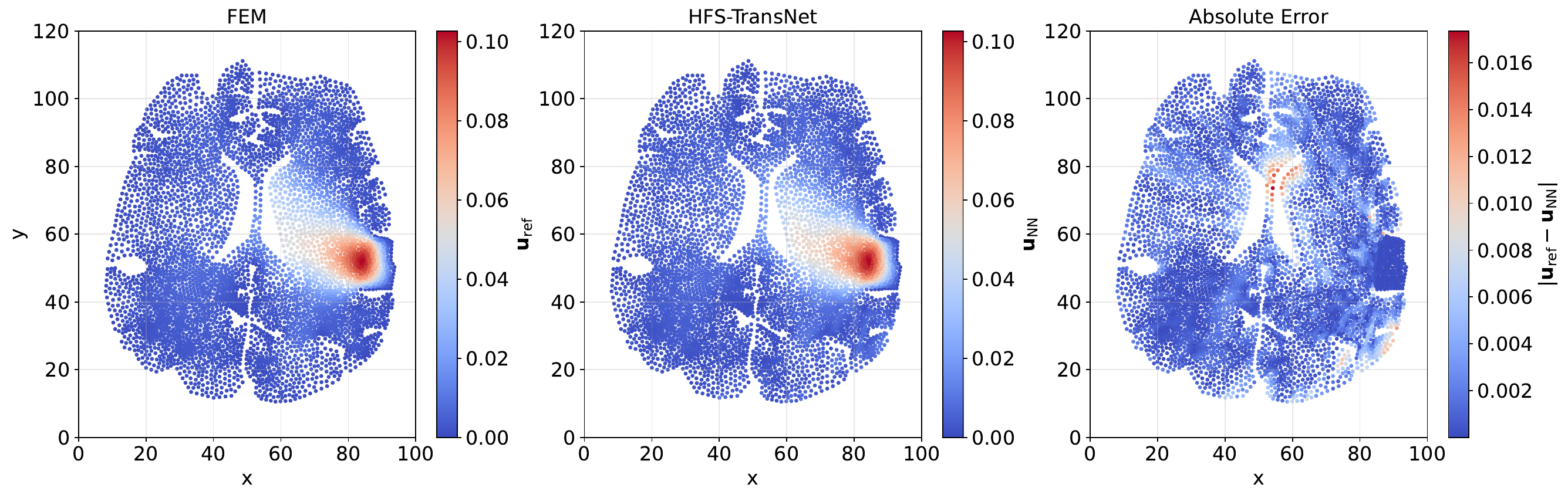}
    \vspace{-0.4cm} 
    \caption{Comparison of pressure (top) and displacement (bottom) at $t = 1.5\,\text{min}$ obtained by FS-FEM (left) and HFS-TransNet (middle), along with their error (right).}
    \label{fig: Brain_HFSTransNet}
\end{figure}

\section{Conclusions}
\label{sec: ne_con}
In this work, we present a hybrid fixed-stress transferable neural network (HFS-TransNet) for solving the Biot model, which governs the coupled interaction between fluids and poroelastic media. 
Approximating the observed data in the initialization domain via the purely data-driven TransNet yields the starting solutions for the subsequent physics-informed module.
These initial solutions are iteratively refined through the fixed-stress splitting-based TransNet.
After each iteration, a greedy update strategy defined on the validation domain is employed to prevent model degradation, and the loop is terminated once the relative tolerance stopping criterion is satisfied.

We establish an error estimate for the physics-informed module embedded with the fixed-stress splitting iteration, which is validated through numerical experiments.
Additionally, these experiments demonstrate that the same hidden-layer pretraining mechanism is transferable across varying physical properties (Poisson's ratios and hydraulic conductivities), distinct boundary conditions (Dirichlet and mixed boundary conditions), diverse geometric domains (regular and irregular), and different time intervals (initialization and validation), thereby substantially reducing the cost of parameter tuning.

Successfully applied to brain edema simulations, HFS-TransNet bridges traditional decoupled iterative algorithms and hybrid scientific machine learning, thus offering a promising tool for complex Biot systems.
Future work will focus on establishing rigorous error estimates for the data-driven module and extending the proposed method to more challenging scenarios, such as 3D brain flow simulations \cite{cai2023iteratively} and multiphase poroelasticity \cite{haghighat2022physics}.

\appendix
\section{Proof of Proposition 3.3}
\label{app: fsTransNet}
\begin{proof}
    First, we subtract \cref{equ: solid force equilibrium} from the $n$-th iteration of \cref{equ: Step 2 PDE NN} and take the derivative with respect to $t$ to obtain
    \begin{equation}
        -\nabla \cdot \sigma(\partial_t e_{\boldsymbol{u}}^n) = -\alpha \nabla \partial_t e_p^n + \alpha \nabla\left(\partial_t p^n - \partial_t p_{\mathrm{NN}}^n\right). \label{equ: thm 2 error 2 partial t}
    \end{equation}
    Multiplying \cref{equ: thm 2 error 2 partial t} by $\partial_t e_{\boldsymbol{u}}^n$ and applying the divergence theorem, we obtain
    \begin{align}
        \begin{split}
        &2\mu \int_0^T \int_{\Omega} \varepsilon\left(\partial_t e_{\boldsymbol{u}}^n\right) : \varepsilon\left(\partial_t e_{\boldsymbol{u}}^n\right) dxdt + \lambda \int_0^T \int_{\Omega} \left(\nabla \cdot \partial_t e_{\boldsymbol{u}}^n\right)^2 dxdt \\
        &= \alpha \int_0^T \int_{\Omega} \left(\partial_t e_p^n\right) \left(\nabla \cdot \partial_t e_{\boldsymbol{u}}^n\right) dxdt + \alpha \int_0^T \int_{\Omega} \left(\nabla\left(\partial_t p^n - \partial_t p_{\mathrm{NN}}^n\right)\right) \cdot \left(\partial_t e_{\boldsymbol{u}}^n\right) dxdt. 
        \end{split}\label{equ: thm 2 int 2}
    \end{align}
    The LHS of \cref{equ: thm 2 int 2} satisfies 
    \begin{equation}
        \text{LHS} \geq \left(\frac{2\mu}{d}+\lambda\right) \int_0^T \int_{\Omega} \left(\nabla \cdot \partial_t e_{\boldsymbol{u}}^n\right)^2 dxdt. \label{equ: thm 2 int 2 LHS}
    \end{equation}
    Using the Cauchy--Schwarz inequality and the Poincaré inequality, for any $\epsilon_1>0$, there exists a $C_p>0$ such that the RHS of \cref{equ: thm 2 int 2} satisfies 
    \begin{align}
        \begin{split}
        \text{RHS} \leq &\frac{\alpha^2}{2\left( \frac{2\mu}{d} + \lambda \right)} \int_0^T \int_{\Omega} \left(\partial_t e_p^n\right)^2 dxdt + \frac{1}{2}\left( \frac{2\mu}{d} + \lambda \right) \int_0^T \int_{\Omega} \left(\nabla \cdot \partial_t e_{\boldsymbol{u}}^n\right)^2 dxdt \\ 
        &+ \frac{\alpha^2}{4\epsilon_1} \int_0^T \int_{\Omega} \left(\nabla \partial_t p^n - \nabla \partial_t p_{\mathrm{NN}}^n\right)^2 dxdt + C_{p} \epsilon_1 \int_0^T \int_{\Omega} \left(\nabla \cdot \partial_t e_{\boldsymbol{u}}^n\right)^2 dxdt. 
        \end{split}\label{equ: thm 2 int 2 RHS}
    \end{align}
    Combining \cref{equ: thm 2 int 2 LHS} and \cref{equ: thm 2 int 2 RHS}, we obtain
    \begin{align}
        \begin{split}
        \left(1-\frac{2C_p \epsilon_1}{\frac{2\mu}{d}+\lambda}\right)\text{LHS} \leq &\frac{\alpha^2}{\frac{2\mu}{d} + \lambda} \int_0^T \int_{\Omega} \left(\partial_t e_p^n\right)^2 dxdt \\
        &+ \frac{\alpha^2}{2\epsilon_1} \int_0^T \int_{\Omega} \left(\nabla \partial_t p^n - \nabla \partial_t p_{\mathrm{NN}}^n\right)^2 dxdt.
        \end{split}
    \end{align}
    Define $C$ such that $\left(1+C\epsilon_1\right)\left(1-\frac{2C_p \epsilon_1}{\frac{2\mu}{d}+\lambda}\right)=1$. 
    Then, by \Cref{equ: ass}, we obtain
    \begin{align}
        \begin{split}
        &2\mu \int_0^T \int_{\Omega} \varepsilon(\partial_t e_{\boldsymbol{u}}^n) : \varepsilon(\partial_t e_{\boldsymbol{u}}^n) dxdt + \lambda \int_0^T \int_{\Omega} \left(\nabla \cdot \partial_t e_{\boldsymbol{u}}^n\right)^2 dxdt \\
        &\leq (1 + C \epsilon_1) \left[ \frac{\alpha^2}{\frac{2\mu}{d} + \lambda} \int_0^T \int_{\Omega} \left(\partial_t e_p^n\right)^2 dxdt + \frac{\alpha^2}{2\epsilon_1} \delta^2 \right]. 
        \end{split}\label{equ: thm 2 ineq 1}
    \end{align}
    
    Similar to \cref{equ: thm 2 error 2 partial t},  we subtract \cref{equ: solid force equilibrium} from the $(n+1)$-th iteration of \cref{equ: Step 2 PDE NN} and take the derivative with respect to $t$ to obtain
    \begin{equation}
        -\nabla \cdot \sigma(\partial_t e_{\boldsymbol{u}}^{n+1}) = -\alpha \nabla \partial_t e_p^{n+1} + \alpha \nabla\left(\partial_t p^{n+1} - \partial_t p_{\mathrm{NN}}^{n+1}\right). \label{equ: thm 2 error 2 partial t n+1}
    \end{equation}
    Subtracting \cref{equ: thm 2 error 2 partial t} from \cref{equ: thm 2 error 2 partial t n+1} yields
    \begin{equation}
        -\nabla \cdot \sigma(\partial_t e_{\boldsymbol{u}}^{n+1} - \partial_t e_{\boldsymbol{u}}^n) = -\alpha \nabla (\partial_t e_p^{n+1} - \partial_t e_p^n) + \alpha \nabla\left(\left(\partial_t p^{n+1} - \partial_t p_{\mathrm{NN}}^{n+1}\right) - \left(\partial_t p^n - \partial_t p_{\mathrm{NN}}^n\right)\right).
    \end{equation}
    Similar to \cref{equ: thm 2 ineq 1}, we obtain
    \begin{align}
        \begin{split}
        &2\mu \int_0^T \int_{\Omega} \varepsilon\left(\partial_t e_{\boldsymbol{u}}^{n+1} - \partial_t e_{\boldsymbol{u}}^n\right) : \varepsilon\left(\partial_t e_{\boldsymbol{u}}^{n+1} - \partial_t e_{\boldsymbol{u}}^n\right) dxdt + \lambda \int_0^T \int_{\Omega} \left(\nabla \cdot \left(\partial_t e_{\boldsymbol{u}}^{n+1} - \partial_t e_{\boldsymbol{u}}^n\right)\right)^2 dxdt \\
        &\leq (1 + C \epsilon_1) \left[ \frac{\alpha^2}{\frac{2\mu}{d} + \lambda} \int_0^T \int_{\Omega} \left(\partial_t e_p^{n+1} - \partial_t e_p^n\right)^2 dxdt + \frac{2\alpha^2}{\epsilon_1} \delta^2 \right]. 
        \end{split}\label{equ: thm 2 ineq 2}
    \end{align}

    Next, subtracting \cref{equ: fluid mass conservation} from the $(n+1)$-th iteration of \cref{equ: Step 1 PDE NN} yields
    \begin{align}
        \begin{split}
        &\left( \frac{1}{M} + \beta_{FS} \right) \partial_t e_p^{n+1} - \nabla \cdot K \nabla e_p^{n+1} \\
        &= -\alpha \nabla \cdot \partial_t e_{\boldsymbol{u}}^n
        + \beta_{FS} \partial_t e_p^n + \alpha \nabla \cdot \left(\partial_t \boldsymbol{u}^n - \partial_t \boldsymbol{u}_{\mathrm{NN}}^n\right) - \beta_{FS} \left(\partial_t p^n - \partial_t p_{\mathrm{NN}}^n\right).  
        \end{split}\label{equ: thm 2 error 1} 
    \end{align}
    Multiplying \cref{equ: thm 2 error 1} by $\partial_t e_p^{n+1}$ and \cref{equ: thm 2 error 2 partial t n+1} by $\partial_t e_{\boldsymbol{u}}^n$ respectively, integrating the results, and summing them, we obtain
    \begin{align}
        \begin{split}
        &\frac{1}{M} \int_0^T \int_{\Omega} \left(\partial_t e_p^{n+1}\right)^2 dxdt + \frac{K}{2} \int_{\Omega} \left(\nabla e_p^{n+1}(T)\right)^2 dx + 2\mu \int_0^T \int_{\Omega} \varepsilon\left(\partial_t e_{\boldsymbol{u}}^{n+1}\right) : \varepsilon\left(\partial_t e_{\boldsymbol{u}}^n\right) dxdt \\
        &\quad + \lambda \int_0^T \int_{\Omega} \left(\nabla \cdot \partial_t e_{\boldsymbol{u}}^{n+1}\right)\left(\nabla \cdot \partial_t e_{\boldsymbol{u}}^n\right) dxdt + \beta_{FS} \int_0^T \int_{\Omega} \left(\partial_t e_p^{n+1}\right)\left(\partial_t e_p^{n+1} - \partial_t e_p^n\right) dxdt \\
        &= -\alpha \int_0^T \int_{\Omega} \left(\nabla \cdot \left(\partial_t \boldsymbol{u}_{\mathrm{NN}}^n - \partial_t \boldsymbol{u}^n\right)\right)\left(\partial_t e_p^{n+1}\right) dxdt + \beta_{FS} \int_0^T \int_{\Omega} \left(\partial_t p_{\mathrm{NN}}^n - \partial_t p^n\right)\left(\partial_t e_p^{n+1}\right) dxdt \\
        &\quad - \alpha \int_0^T \int_{\Omega} \left(\nabla \left(\partial_t p_{\mathrm{NN}}^{n+1} - \partial_t p^{n+1}\right)\right)\left(\partial_t e_{\boldsymbol{u}}^n\right) dxdt. 
        \end{split}\label{equ: thm 2 total}
    \end{align}
    Using the following algebraic identities to the LHS of \cref{equ: thm 2 total},
    \begin{equation}
        \left\{
            \begin{aligned}
            (\xi, \xi - \eta) &= \frac{1}{2}\left( (\xi, \xi) - (\eta, \eta) + (\xi - \eta, \xi - \eta) \right), \\
            (\xi, \eta) &= \frac{1}{4}\left( (\xi + \eta, \xi + \eta) - (\xi - \eta, \xi - \eta) \right).
            \end{aligned} \label{equ: algebraic identities}
        \right.
    \end{equation}
    we obtain
    \begin{align}
        \begin{split}
        \text{LHS} &\geq \left( \frac{1}{M} + \frac{\beta_{FS}}{2} \right) \int_0^T \int_{\Omega} \left(\partial_t e_p^{n+1}\right)^2 dxdt - \frac{\beta_{FS}}{2} \int_0^T \int_{\Omega} \left(\partial_t e_p^n\right)^2 dxdt \\
        &\quad+ \frac{\beta_{FS}}{2} \int_0^T \int_{\Omega} \left(\partial_t e_p^{n+1} - \partial_t e_p^n\right)^2 dxdt - \frac{\lambda}{4} \int_0^T \int_{\Omega} \left(\nabla \cdot \left(\partial_t e_{\boldsymbol{u}}^{n+1} - \partial_t e_{\boldsymbol{u}}^n\right)\right)^2 dxdt \\
        &\quad - \frac{\mu}{2} \int_0^T \int_{\Omega} \varepsilon\left(\partial_t e_{\boldsymbol{u}}^{n+1} - \partial_t e_{\boldsymbol{u}}^n\right) : \varepsilon\left(\partial_t e_{\boldsymbol{u}}^{n+1} - \partial_t e_{\boldsymbol{u}}^n\right) dxdt. 
        \end{split}\label{equ: thm 2 total LHS} 
    \end{align}
    Based on \cref{equ: thm 2 total LHS} and \cref{equ: thm 2 ineq 2}, we obtain
    \begin{align}
        \begin{split}
        \text{LHS} &\geq \left( \frac{1}{M} + \frac{\beta_{FS}}{2} \right) \int_0^T \int_{\Omega} \left(\partial_t e_p^{n+1}\right)^2 dxdt - \frac{\beta_{FS}}{2} \int_0^T \int_{\Omega} \left(\partial_t e_p^n\right)^2 dxdt \\
        &\quad + \left( \frac{\beta_{FS}}{2} - \frac{(1 + C \epsilon_1)}{4} \frac{\alpha^2}{\frac{2\mu}{d} + \lambda} \right) \int_0^T \int_{\Omega} \left(\partial_t e_p^{n+1} - \partial_t e_p^n\right)^2 dxdt - \frac{(1+C \epsilon_1)}{2}\frac{\alpha^2}{\epsilon_1}\delta^2. 
        \end{split}\label{equ: thm 2 total LHS simp}
    \end{align}
    Using the Cauchy--Schwarz inequality and the Poincaré inequality, for any $\epsilon_2 > 0$, the RHS of \cref{equ: thm 2 total} satisfies
    \begin{align}
        \begin{split}
        \text{RHS} &\leq \frac{\alpha^2}{4\epsilon_2} \int_0^T \int_{\Omega} \left(\nabla \cdot \left(\partial_t \boldsymbol{u}_{\mathrm{NN}}^n - \partial_t \boldsymbol{u}^n\right)\right)^2 dxdt + \epsilon_2 \int_0^T \int_{\Omega} \left(\partial_t e_p^{n+1}\right)^2 dxdt \\
        &\quad + \frac{\beta_{FS}}{4\epsilon_2} \int_0^T \int_{\Omega} \left(\partial_t p_{\mathrm{NN}}^n - \partial_t p^n\right)^2 dxdt + \beta_{FS}\epsilon_2 \int_0^T \int_{\Omega} \left(\partial_t e_p^{n+1}\right)^2 dxdt \\
        &\quad + \frac{\alpha^2}{4\epsilon_2} \int_0^T \int_{\Omega} \left(\nabla \left(\partial_t p_{\mathrm{NN}}^{n+1} - \partial_t p^{n+1}\right)\right)^2 dxdt + C_p \epsilon_2 \int_0^T \int_{\Omega} \left(\nabla \cdot \partial_t e_{\boldsymbol{u}}^n\right)^2 dxdt. 
        \end{split}\label{equ: thm 2 total RHS}
    \end{align}
    Applying \Cref{equ: ass} and \cref{equ: thm 2 ineq 1}, we obtain
    \begin{align}
        \begin{split}
        \text{RHS} &\leq \left(1+\beta_{FS}\right)\epsilon_2 \int_0^T \int_{\Omega} \left(\partial_t e_p^{n+1}\right)^2 dxdt + C_p \epsilon_2 \int_0^T \int_{\Omega} \left(\nabla \cdot \partial_t e_{\boldsymbol{u}}^n\right)^2 dxdt + \frac{2\alpha^2 + \beta_{FS}}{4\epsilon_2} \delta^2 \\
        &\leq \left(1+\beta_{FS}\right)\epsilon_2 \int_0^T \int_{\Omega} \left(\partial_t e_p^{n+1}\right)^2 dxdt \\
        &\quad+ \frac{C_p \epsilon_2}{\lambda} \left(1+C\epsilon_1\right) \left[ \frac{\alpha^2}{\frac{2\mu}{d} + \lambda} \int_0^T \int_{\Omega} \left(\partial_t e_p^n\right)^2 dxdt + \frac{\alpha^2}{2\epsilon_1} \delta^2 \right] + \frac{2\alpha^2 + \beta_{FS}}{4\epsilon_2} \delta^2. 
        \end{split}\label{equ: thm 2 total RHS simp}
    \end{align}
    Setting $\beta_{FS} \geq \frac{\left(1+C\epsilon_1\right)\alpha^2}{2\left(\frac{2\mu}{d}+\lambda\right)}$ and combining \cref{equ: thm 2 total LHS simp} and \cref{equ: thm 2 total RHS simp}, we have
    \begin{align}
        \begin{split}
        &\left( \frac{1}{M} + \frac{\beta_{FS}}{2} - \left(1+\beta_{FS}\right)\epsilon_2 \right) \int_0^T \int_{\Omega} \left(\partial_t e_p^{n+1}\right)^2 dxdt \\
        &\leq \left( \frac{\beta_{FS}}{2} + C_1(\epsilon_1, \epsilon_2) \right) \int_0^T \int_{\Omega} \left(\partial_t e_p^n\right)^2 dxdt + C_2(\epsilon_1, \epsilon_2) \delta^2, 
        \end{split}\label{equ: thm 2 result}
    \end{align}
    where
    \begin{equation}
        \left\{
        \begin{aligned}
            C_1(\epsilon_1, \epsilon_2) &= \frac{C_p\epsilon_2(1+C\epsilon_1)\alpha^2}{\lambda\left(\frac{2\mu}{d}+\lambda\right)}, \\
            C_2(\epsilon_1, \epsilon_2) &= \frac{2\alpha^2 + \beta_{FS}}{4\epsilon_2} + \left( \frac{C\epsilon_1 + 1}{2} \right)\left(\frac{C_p\epsilon_2}{\lambda}+1\right) \frac{\alpha^2}{\epsilon_1}. \nonumber
        \end{aligned}
        \right.
    \end{equation}

    The contraction coefficient $\tilde{L}$ is defined as follows:
    \begin{equation}
        \tilde{L} = \frac{\frac{\beta_{FS}}{2} + C_1(\epsilon_1, \epsilon_2)}{\frac{1}{M} + \frac{\beta_{FS}}{2} - \left(1+\beta_{FS}\right)\epsilon_2}. \nonumber
    \end{equation}
    The parameters $\epsilon_1$ and $\epsilon_2$ are chosen sufficiently small to guarantee that $\tilde{L}<1$. Then, \cref{equ: thm 2 result} can be rewritten as
    \begin{equation}
        \int_0^T \int_{\Omega} \left(\partial_t e_p^{n+1}\right)^2 dxdt \leq \tilde{L} \int_0^T \int_{\Omega} \left(\partial_t e_p^n\right)^2 dxdt + C\delta^2. \label{equ: thm 2 result simp}
    \end{equation}
    Furthermore, \cref{equ: thm 2 result simp} yields the limit supremum
    \begin{align}
        &\limsup_{n \to \infty} \int_0^T \int_{\Omega} \left(\partial_t e_p^{n+1}\right)^2 dxdt \nonumber\\
        &\leq \limsup_{n \to \infty} \left[ \tilde{L}^{n+1} \int_0^T \int_{\Omega} \left(\partial_t e_p^0\right)^2 dxdt + \left(\tilde{L}^n + \cdots + 1\right)C\delta^2 \right] = C \frac{\delta^2}{1-\tilde{L}}. \nonumber
    \end{align}
    This completes the proof.
\end{proof}

\section{Implementation details} 
\label{app: para}
This section outlines the detailed information used in the numerical experiments.

In this paper, for brain edema simulations, the spatial collocation and test points coincide with the mesh nodes utilized by the FS-FEM, with temporal snapshots being uniformly sampled. 
For the other experiments, the collocation points are uniformly selected across the space-time domain, while the test points are uniformly distributed in the spatial domain at a specific time instant.

Regarding the data-driven module, the observed data are generated by the exact solutions in Test Cases 1 and 2, while they are obtained by the FS-FEM in Barry--Mercer problem and the brain edema simulation.
Conversely, no observed data are used in the Mandel problem.

In all experiments, we set: 
\begin{itemize}
    \item The parameter $\beta_{FS} = \alpha^2/(2\left(2\mu/d+\lambda\right))$.
    \item The maximum number of fixed-stress splitting iterations $L=10$.
    \item The relative tolerance for HFS-TransNet $\text{Tol}=10^{-8}$.
\end{itemize}

The main neural network parameters used in the experiments are presented below:
\begin{itemize}
    \item $\boldsymbol{x}_c$, $R$: the center and radius of the hidden-layer neuron distribution region $B_R(\boldsymbol{x}_c)$.
    \item $\gamma$, N            : the shape parameter and the number of hidden-layer neurons.
    \item $\text{D}_{\text{init}}$, $\text{D}_{\text{val}}$: the sampling domains for initialization and validation within the data-driven module.
    \item $\text{D}_{\text{train}}$: the sampling domain for iteration within the physics-informed module.
    \item $\text{D}_{\text{test}}$: the sampling domain for accuracy testing.
    \item $\text{P}_{\text{init}}$, $\text{P}_{\text{val}}$: the number of collocation points used for initialization and validation within the data-driven module.
    \item $\text{P}_{\text{in}}$, $\text{P}_{\text{bd}}$, $\text{P}_{\text{ic}}$: the number of interior, boundary, and initial collocation points for iteration within the physics-informed module.
    \item $\text{P}_{\text{test}}$: the number of collocation points for accuracy testing.
\end{itemize}
The network parameters for Test Cases 1 and 2 are provided in \cref{tab: net_para_case1_case2}, while the parameters for the Barry--Mercer, Mandel, and brain edema problems are given in \cref{tab: hfs_para}.
In the brain edema simulations, the superscripts $n$ and $i$ refer to the normal and injured regions, while the subscript ``co'' identifies the interface or shared collocation points between $\Omega^n$ and $\Omega^i$.
\begin{table}[htbp]
    \footnotesize
    \caption{Network parameters for TransNet, FS-TransNet, and HFS-TransNet in Test Case 1 and 2.}
    \label{tab: net_para_case1_case2}
    \begin{center}
    \begin{tabular}{l|l l|l l l}
    \hline
    \multirow{2}*{Symbol} & \multicolumn{2}{c|}{Test Case 1}& \multicolumn{3}{c}{Test Case 2} \\
    \cline{2-3} \cline{4-6}
    & FS-TransNet & HFS-TransNet & TransNet & FS-TransNet & HFS-TransNet\\
    \hline
    $\boldsymbol{x}_c$  & [0, 0, 0]         & [0, 0, 0] 
    & [0.5, 0.5, 0.5]   & [0.5, 0.5, 0.5]   & [0.5, 0.5, 0.5] \\
    $R$                 & 1.75              & 1.75 
    & 0.9               & 0.9               & 0.9 \\
    $\gamma$            & 1                 & 1 
    & 2                 & 2                 & 2 \\
    N                   & 640               & 640 
    & 1280              & 1280              & 1280 \\
    $\text{D}_{\text{init}}$   & -     & $[0,1]^2 \times [0, 0.2]$ 
    &-                         & -     & $[0,1]^2 \times [0, 0.075]$ \\ 
    $\text{D}_{\text{val}}$    & -     & $[0,1]^2 \times (0.2, 0.4]$ 
    &-                         & -     & $[0,1]^2 \times (0.075, 0.15]$\\ 
    $\text{D}_{\text{train}}$  & $[0,1]^2 \times [0, 1]$ & $[0,1]^2 \times [0, 1]$ 
    &$[0,1]^2 \times [0, 1]$   & $[0,1]^2 \times [0, 1]$ & $[0,1]^2 \times [0, 1]$ \\ 
    $\text{D}_{\text{test}}$   & $[0,1]^2 \times \{0.5\}$ & $[0,1]^2 \times \{0.5\}$ 
    &$[0,1]^2 \times \{0.2\}$  & $[0,1]^2 \times \{0.2\}$ & $[0,1]^2 \times \{0.2\}$ \\ 
    $\text{P}_{\text{init}}$   & -     & 2205 
    &-                         & -     & 6724 \\ 
    $\text{P}_{\text{val}}$    & -     & 8405 
    &-                         & -     & 26244 \\ 
    $\text{P}_{\text{in}}$     & 9261  & 9261 
    &68921                     & 68921 & 68921 \\ 
    $\text{P}_{\text{bd}}$     & 6724  & 6724 
    &16900                     & 16900 & 16900 \\ 
    $\text{P}_{\text{ic}}$     & 1681  & 1681 
    &4225                      & 4225  & 4225 \\ 
    $\text{P}_{\text{test}}$   & 1681  & 1681 
    &6561                      & 6561  & 6561 \\
    \hline
    \end{tabular}
    \end{center}
\end{table}
\begin{table}[htbp]
    \footnotesize
    \caption{Network parameters for HFS-TransNet in the Barry--Mercer problem, Mandel problem, and brain edema simulations.}
    \label{tab: hfs_para}
    \begin{center}
    \begin{tabular}{l|l|l|l l}
    \hline
    \multirow{2}*{Symbol} & \multirow{2}*{Barry--Mercer} & \multirow{2}*{Mandel} &\multicolumn{2}{c}{brain edema simulations} \\
    \cline{4-5}
    & & & Normal region & Injured region\\
    \hline
    $\boldsymbol{x}_c$  & [0.5, 0.5, 0.8] & [0.5, 0.5, 0.5] & [80, 50, 0.75] & [90, 50, 0.75] \\
    $R$                 & 1.1             & 0.75            & 10             & 10             \\
    $\gamma$            & 6               & 10              & 1              & 1              \\
    N                   & 1280            & 1280            & 2048           & 2048           \\
    $\text{D}_{\text{init}}$  & $[0,1]^2 \times [0.08, 0.4]$ & - & $\Omega^n \times [0, 1]$   & $\Omega^i \times [0, 1]$ \\ 
    $\text{D}_{\text{val}}$   & $[0,1]^2 \times (0.4, 0.8]$  & - & $\Omega^n \times (1, 1.1]$ & $\Omega^i \times (1, 1.1]$ \\ 
    $\text{D}_{\text{co}}$    & - & - &$(\Omega^n \cap \Omega^i) \times [0, 1.5]$ & $(\Omega^n \cap \Omega^i) \times [0, 1.5]$ \\
    $\text{D}_{\text{train}}$ & $[0,1]^2 \times [0, 1.6]$    & $[0,1]^2 \times [0.001, 1]$ & $\Omega^n \times [0, 1.5]$ & $\Omega^i \times [0, 1.5]$ \\ 
    $\text{D}_{\text{test}}$  & $[0,1]^2 \times \{\pi/2\}$   & $[0,1]^2 \times \{1\}$ & $(\Omega^n \cup \Omega^i) \times \{1.5\}$ & $(\Omega^n \cup \Omega^i) \times \{1.5\}$ \\ 
    $\text{P}_{\text{init}}$  & 2565 & -     & 470500 & 20100 \\ 
    $\text{P}_{\text{val}}$   & 2565 & -     & 47050  & 2010 \\ 
    $\text{P}_{\text{co}}$    & -    & -     & 2000   & 2000  \\
    $\text{P}_{\text{in}}$    & 9261 & 68921 & 470500 & 20100  \\ 
    $\text{P}_{\text{bd}}$    & 6724 & 26244 & 52900  & 2900  \\ 
    $\text{P}_{\text{ic}}$    & 1681 & 6561  & 4705   & 201 \\ 
    $\text{P}_{\text{test}}$  & 1681 & 16641 & 4705   & 201 \\
    \hline
    \end{tabular}
    \end{center}
\end{table}

The network methods (FS-PINN, TransNet, FS-TransNet, and HFS-TransNet) are implemented using PyTorch with an NVIDIA GeForce RTX 4090 GPU, while the FS-FEM method is implemented using FreeFem++ with an Intel(R) Core(TM) i5-9300H CPU @ 2.40GHz.

\section*{Acknowledgments}
L. Zhu's work was partially supported by National Natural Science Foundation of China 
under grant number 11871091, 12571580.

\section*{Data availability}
All the data and source codes to reproduce the results will be made available on GitHub at  \url{https://github.com/ZhequanShen/HFS-TransNet}.

\bibliographystyle{siamplain}
\bibliography{references}

@article{biot1941general,
  title={General theory of three-dimensional consolidation},
  author={Biot, Maurice A},
  journal={J. Appl. Phys.},
  volume={12},
  number={2},
  pages={155--164},
  year={1941},
  publisher={American Institute of Physics},
  doi = {10.1063/1.1712886}
}

@article{vzenivsek1984existence,
  title={The existence and uniqueness theorem in {Biot's} consolidation theory},
  author={{\v{Z}}en{\'\i}{\v{s}}ek, Alexander},
  journal={Appl. Math.},
  volume={29},
  number={3},
  pages={194--211},
  year={1984},
  publisher={Institute of Mathematics, Academy of Sciences of the Czech Republic},
  doi = {10.21136/AM.1984.104085}
}

@article{showalter2000diffusion,
  title={Diffusion in poro-elastic media},
  author={Showalter, Ralph E},
  journal={J. Math. Anal. Appl.},
  volume={251},
  number={1},
  pages={310--340},
  year={2000},
  publisher={Elsevier},
  doi = {10.1006/jmaa.2000.7048}
}

@article{huang2022generalized,
  title={Generalized effective {Biot} theory and seismic wave propagation in anisotropic, poroviscoelastic media},
  author={Huang, Xingguo and Greenhalgh, Stewart and Han, Li and Liu, Xu},
  journal={J. Geophys. Res. Solid Earth},
  volume={127},
  number={3},
  pages={e2021JB023590},
  year={2022},
  publisher={Wiley Online Library},
  doi = {10.1029/2021JB023590}
}

@article{ju2020parameter,
  title={Parameter-robust multiphysics algorithms for {Biot} model with application in brain edema simulation},
  author={Ju, Guoliang and Cai, Mingchao and Li, Jingzhi and Tian, Jing},
  journal={Math. Comput. Simulation},
  volume={177},
  pages={385--403},
  year={2020},
  publisher={Elsevier},
  doi = {10.1016/j.matcom.2020.04.027}
}

@article{wang2025mixed,
  title={A mixed virtual element method for the four-field poroelasticity problem on polygonal meshes and its simulation in brain edema},
  author={Wang, Gang and Cai, Mingchao and Dong, Xiaojing and He, Yinnian},
  journal={Numer. Methods Partial Differential Equations},
  volume={41},
  number={6},
  pages={e70052},
  year={2025},
  publisher={Wiley Online Library},
  doi = {10.1002/num.70052}
}

@article{chen2023physical,
  title={Physical information neural networks for 2D and 3D nonlinear {Biot} model and simulation on the pressure of brain},
  author={Chen, Hao and Ge, Zhihao},
  journal={J. Comput. Phys.},
  volume={490},
  pages={112309},
  year={2023},
  publisher={Elsevier},
  doi = {10.1016/j.jcp.2023.112309}
}

@article{kim20111591,
  title = {Stability and convergence of sequential methods for coupled flow and geomechanics: Fixed-stress and fixed-strain splits},
  journal = {Comput. Methods Appl. Mech. Engrg.},
  volume = {200},
  number = {13},
  pages = {1591-1606},
  year = {2011},
  author = {J. Kim and H.A. Tchelepi and R. Juanes},
  doi = {10.1016/j.cma.2010.12.022}
}

@article{mikelic2013convergence,
  title={Convergence of iterative coupling for coupled flow and geomechanics},
  author={Mikeli{\'c}, Andro and Wheeler, Mary F},
  journal={Comput. Geosci.},
  volume={17},
  number={3},
  pages={455--461},
  year={2013},
  publisher={Springer},
  doi = {10.1007/s10596-012-9318-y}
}

@article{cai2023combination,
  title={A combination of physics-informed neural networks with the fixed-stress splitting iteration for solving {Biot's} model},
  author={Cai, Mingchao and Gu, Huipeng and Hong, Pengxiang and Li, Jingzhi},
  journal={Front. Appl. Math. Stat.},
  volume={9},
  pages={1206500},
  year={2023},
  publisher={Frontiers Media SA},
  doi = {10.3389/fams.2023.1206500}
}

@article{phillips2007coupling,
  title={A coupling of mixed and continuous Galerkin finite element methods for poroelasticity I: the continuous in time case},
  author={Phillips, Phillip Joseph and Wheeler, Mary F},
  journal={Comput. Geosci.},
  volume={11},
  number={2},
  pages={131--144},
  year={2007},
  publisher={Springer},
  doi = {10.1007/s10596-007-9045-y}
}

@article{yi2017study,
  title={A study of two modes of locking in poroelasticity},
  author={Yi, Son-Young},
  journal={SIAM J. Numer. Anal.},
  volume={55},
  number={4},
  pages={1915--1936},
  year={2017},
  publisher={SIAM},
  doi = {10.1137/16M1056109}
}

@article{lee2017parameter,
  title={Parameter-robust discretization and preconditioning of {Biot's} consolidation model},
  author={Lee, Jeonghun J and Mardal, Kent-Andre and Winther, Ragnar},
  journal={SIAM J. Sci. Comput.},
  volume={39},
  number={1},
  pages={A1--A24},
  year={2017},
  publisher={SIAM},
  doi = {10.1137/15M1029473}
}

@article{mora2026118447,
  title = {A virtual element method for the {Biot-Brinkman} equations using Nitsche’s technique},
  journal = {Comput. Methods Appl. Mech. Engrg.},
  volume = {448},
  pages = {118447},
  year = {2026},
  author = {David Mora and Jesus Vellojin and Nitesh Verma},
  doi = {10.1016/j.cma.2025.118447}
}

@article{feng2018analysis,
  title={Analysis of a multiphysics finite element method for a poroelasticity model},
  author={Feng, Xiaobing and Ge, Zhihao and Li, Yukun},
  journal={IMA J. Numer. Anal.},
  volume={38},
  number={1},
  pages={330--359},
  year={2018},
  publisher={Oxford University Press},
  doi = {10.1093/imanum/drx003}
}

@article{oyarzua2016locking,
  title={Locking-free finite element methods for poroelasticity},
  author={Oyarz{\'u}a, Ricardo and Ruiz-Baier, Ricardo},
  journal={SIAM J. Numer. Anal.},
  volume={54},
  number={5},
  pages={2951--2973},
  year={2016},
  publisher={SIAM},
  doi = {10.1137/15M1050082}
}

@article{kim2011stability,
  title={Stability, accuracy, and efficiency of sequential methods for coupled flow and geomechanics},
  author={Kim, Jihoon and Tchelepi, Hamdi A and Juanes, Ruben},
  journal={SPE J.},
  volume={16},
  number={02},
  pages={249--262},
  year={2011},
  publisher={OnePetro},
  doi = {10.2118/119084-PA}
}

@article{storvik2019optimization,
  title={On the optimization of the fixed-stress splitting for {Biot's} equations},
  author={Storvik, Erlend and Both, Jakub W and Kumar, Kundan and Nordbotten, Jan M and Radu, Florin A},
  journal={Internat. J. Numer. Methods Engrg.},
  volume={120},
  number={2},
  pages={179--194},
  year={2019},
  publisher={Wiley Online Library},
  doi = {10.1002/nme.6130}
}

@article{chen2025due,
  title={{DUE}: A deep learning framework and library for modeling unknown equations},
  author={Chen, Junfeng and Wu, Kailiang and Xiu, Dongbin},
  journal={SIAM Rev.},
  volume={67},
  number={4},
  pages={873--902},
  year={2025},
  publisher={SIAM},
  doi = {10.1137/24M1671827}
}

@article{raissi2019physics,
  title={Physics-informed neural networks: A deep learning framework for solving forward and inverse problems involving nonlinear partial differential equations},
  author={Raissi, Maziar and Perdikaris, Paris and Karniadakis, George E},
  journal={J. Comput. Phys.},
  volume={378},
  pages={686--707},
  year={2019},
  publisher={Elsevier},
  doi = {10.1016/j.jcp.2018.10.045}
}

@article{haghighat2022physics,
  title={Physics-informed neural network simulation of multiphase poroelasticity using stress-split sequential training},
  author={Haghighat, Ehsan and Amini, Danial and Juanes, Ruben},
  journal={Comput. Methods Appl. Mech. Engrg.},
  volume={397},
  pages={115141},
  year={2022},
  publisher={Elsevier},
  doi = {10.1016/j.cma.2022.115141}
}

@article{liverani2025hyco,
  title={{HYCO}: Hybrid-cooperative learning for data-driven pde modeling},
  author={Liverani, Lorenzo and Steynberg, Matthys and Zuazua, Enrique},
  journal={arXiv preprint arXiv:2509.14123},
  year={2025},
  doi = {10.48550/arXiv.2509.14123}
}

@article{Quarteroni2025com,
  author = {Quarteroni, Alfio and Gervasio, Paola and Regazzoni, Francesco},
  title = {Combining physics-based and data-driven models: advancing the frontiers of research with scientific machine learning},
  journal = {Math. Models Methods Appl. Sci.},
  volume = {35},
  number = {04},
  pages = {905-1071},
  year = {2025},
  doi = {10.1142/S0218202525500125}
}

@article{cai2020deep,
  title={Deep least-squares methods: An unsupervised learning-based numerical method for solving elliptic PDEs},
  author={Cai, Zhiqiang and Chen, Jingshuang and Liu, Min and Liu, Xinyu},
  journal={J. Comput. Phys.},
  volume={420},
  pages={109707},
  year={2020},
  publisher={Elsevier},
  doi = {10.1016/j.jcp.2020.109707}
}

@article{dong2021local,
  title={Local extreme learning machines and domain decomposition for solving linear and nonlinear partial differential equations},
  author={Dong, Suchuan and Li, Zongwei},
  journal={Comput. Methods Appl. Mech. Engrg.},
  volume={387},
  pages={114129},
  year={2021},
  publisher={Elsevier},
  doi = {10.1016/j.cma.2021.114129}
}

@article{dong2022computing,
  title={On computing the hyperparameter of extreme learning machines: Algorithm and application to computational PDEs, and comparison with classical and high-order finite elements},
  author={Dong, Suchuan and Yang, Jielin},
  journal={J. Comput. Phys.},
  volume={463},
  pages={111290},
  year={2022},
  publisher={Elsevier},
  doi = {10.1016/j.jcp.2022.111290}
}

@article{chen2022bridging,
  title={Bridging traditional and machine learning-based algorithms for solving PDEs: the random feature method},
  author={Chen, Jingrun and Chi, Xurong and E, Weinan and Yang, Zhouwang},
  journal={J. Mach. Learn.},
  volume={1},
  pages={268--98},
  year={2022},
  doi = {10.4208/jml.220726}
}

@article{li2025local,
  title={Local randomized neural networks with finite difference methods for interface problems},
  author={Li, Yunlong and Wang, Fei},
  journal={J. Comput. Phys.},
  volume={529},
  pages={113847},
  year={2025},
  publisher={Elsevier},
  doi = {10.1016/j.jcp.2025.113847}
}

@article{shang2025overlapping,
  title={Overlapping Schwarz preconditioners for randomized neural networks with domain decomposition},
  author={Shang, Yong and Heinlein, Alexander and Mishra, Siddhartha and Wang, Fei},
  journal={Comput. Methods Appl. Mech. Engrg.},
  volume={442},
  pages={118011},
  year={2025},
  publisher={Elsevier},
  doi = {10.1016/j.cma.2025.118011}
}

@article{weng2026deep,
  title={{Deep Collocation Method}: A Framework for Solving PDEs Using Neural Networks with Error Control},
  author={Weng, Mingxing and Mao, Zhiping and Shen, Jie},
  journal={SIAM J. Sci. Comput.},
  volume={48},
  number={1},
  pages={C77--C102},
  year={2026},
  publisher={SIAM},
  doi = {10.1137/25M1739753}
}

@article{zhang2024transferable,
  title={Transferable Neural Networks for Partial Differential Equations},
  author={Zhang, Zezhong and Bao, Feng and Ju, Lili and Zhang, Guannan},
  journal={J. Sci. Comput.},
  volume={99},
  number={1},
  pages={2},
  year={2024},
  publisher={Springer},
  doi = {10.1007/s10915-024-02463-y}
}

@article{lu2025multiple,
  title={A multiple transferable neural network method with domain decomposition for elliptic interface problems},
  author={Lu, Tianzheng and Ju, Lili and Zhu, Liyong},
  journal={J. Comput. Phys.},
  pages={113902},
  year={2025},
  publisher={Elsevier},
  doi = {10.1016/j.jcp.2025.113902}
}

@article{yan2026time,
  title={Time-stepping transferable neural networks for heat conduction with moving heat sources in additive manufacturing},
  author={Yan, Linhuan and Hu, Boya and Guan, Xiaofei and Yang, Zihao and Ma, Liang},
  journal={Adv. Eng. Inform.},
  volume={73},
  pages={104601},
  year={2026},
  publisher={Elsevier},
  doi = {10.1016/j.aei.2026.104601}
}

@article{xu2026weak,
  title={{Weak TransNet}: A Petrov-Galerkin Based Neural Network Method for Solving Elliptic PDEs},
  author={Xu, Zhihang and Wang, Min and Wang, Zhu},
  journal={J. Sci. Comput.},
  volume={107},
  number={2},
  pages={60},
  year={2026},
  publisher={Springer},
  doi = {10.1007/s10915-026-03273-0}
}

@article{shen2026matched,
  title={Matched Asymptotic Expansions-Based Transferable Neural Networks for Singular Perturbation Problems},
  author={Shen, Zhequan and Ju, Lili and Zhu, Liyong},
  journal={Internat. J. Numer. Methods Engrg.},
  volume={127},
  number={6},
  pages={e70307},
  year={2026},
  publisher={Wiley Online Library},
  doi = {10.1002/nme.70307}
}

@article{cheng2026generalized,
  title={Generalized Transferable Neural Networks for Steady-State Partial Differential Equations},
  author={Cheng, Tao and Ju, Lili and Qiao, Zhonghua and Zhang, Xiaoping},
  journal={arXiv preprint arXiv:2604.03020},
  year={2026},
  doi = {10.48550/arXiv.2604.03020}
}

@article{hornik1991approximation,
  title={Approximation capabilities of multilayer feedforward networks},
  author={Hornik, Kurt},
  journal={Neural Netw.},
  volume={4},
  number={2},
  pages={251--257},
  year={1991},
  publisher={Elsevier},
  doi = {10.1016/0893-6080(91)90009-T}
}

@article{barry1999exact,
  title={Exact solutions for two-dimensional time-dependent flow and deformation within a poroelastic medium},
  author={Barry, SI and Mercer, GN},
  journal={J. Appl. Mech.},
  volume={66},
  number={2},
  pages={536--540},
  year={1999},
  publisher={The American Society of Mechanical Engineers},
  doi = {10.1115/1.2791080}
}

@article{phillips2009overcoming,
  title={Overcoming the problem of locking in linear elasticity and poroelasticity: an heuristic approach},
  author={Phillips, Phillip Joseph and Wheeler, Mary F},
  journal={Comput. Geosci.},
  volume={13},
  number={1},
  pages={5--12},
  year={2009},
  publisher={Springer},
  doi = {10.1007/s10596-008-9114-x}
}

@article{mandel1953consolidation,
  title={Consolidation des sols ({\'e}tude math{\'e}matique)},
  author={Mandel, Jean},
  journal={Geotechnique},
  volume={3},
  number={7},
  pages={287--299},
  year={1953},
  publisher={Thomas Telford Ltd},
  doi = {10.1680/geot.1953.3.7.287}
}

@article{abousleiman1996mandel,
  title={Mandel's problem revisited},
  author={Abousleiman, Y and Cheng, AH-D and Cui, L and Detournay, E and Roegiers, J-C},
  journal={Geotechnique},
  volume={46},
  number={2},
  pages={187--195},
  year={1996},
  publisher={Thomas Telford Ltd},
  doi = {10.1680/geot.1996.46.2.187}
}

@article{cai2023iteratively,
  title={An Iteratively Decoupled Algorithm for Multiple-Network Poroelastic Model with Applications in Brain Flow Simulations},
  author={Cai, Mingchao and Lei, Meng and Li, Jingzhi and Sun, Jiaao and Wang, Feng},
  journal={arXiv preprint arXiv:2310.15457},
  year={2023},
  doi = {10.48550/arXiv.2310.15457}
}

@article{chu2026block,
    author = {Chu, Hanyu and Pavarino, Luca Franco and Zampini, Stefano},
    title = {{Block BDDC/FETI-DP} Preconditioners for Three-Field Mixed Finite Element Discretizations of {Biot’s} Consolidation Model},
    journal = {SIAM J. Sci. Comput.},
    volume = {48},
    number = {3},
    pages = {A1710-A1733},
    year = {2026},
    doi = {10.1137/25M1748652}
}
\end{document}